\documentclass[11pt]{article}
\usepackage{fullpage}

\usepackage{bbm}
\usepackage{amssymb}
\usepackage{amsmath}
\usepackage{amsthm}
\usepackage{amsfonts}
\usepackage{graphicx}
\usepackage{url}
\usepackage{hyperref}
\usepackage{color}
\usepackage{natbib}

\newtheorem{theorem}{Theorem}[section]
\newtheorem{mechanism}[theorem]{Mechanism}
\newtheorem{definition}[theorem]{Definition}
\newtheorem{proposition}[theorem]{Proposition}
\newtheorem{lemma}[theorem]{Lemma}

\newtheorem{corollary}[theorem]{Corollary}
\newtheorem{example}[theorem]{Example}
\newtheorem{construction}[theorem]{Construction}

\newcommand{\cI}{\mathcal{I}}
\newcommand{\cG}{\mathcal{G}}
\newcommand{\cM}{\mathcal{M}}
\newcommand{\reals}{\mathbb{R}}
\newcommand{\Pb}{\mathbb{P}}
\newcommand{\E}{\mathbb{E}}

\newcommand{\Var}{\mathbb{V}}

\newcommand{\AutoAdjust}[3]{\mathchoice{ \left #1 #2  \right #3}{#1 #2 #3}{#1 #2 #3}{#1 #2 #3} }
\newcommand{\Xcomment}[1]{{}}

\newcommand{\InBrackets}[1]{\AutoAdjust{[}{#1}{]}}
\newcommand{\Ex}[2][]{\operatorname{\mathbf E}_{#1}\InBrackets{#2}}

\newcommand{\given}{\;\mid\;}

\newcommand{\itemtype}{i}
	\newcommand{\Alphabet}{\Sigma}
	\newcommand{\sig}{s}

	\newcommand{\partsig}{\bar x}

\title{Third-Party Data Providers Ruin Simple Mechanisms}

\author{Yang Cai\thanks{Yale University. Email: yang.cai@yale.edu.}
\and
Federico Echenique\thanks{California Institute of Technology. Email: fede@hss.caltech.edu.}
\and
Hu Fu\thanks{University of British Columbia. Email: hufu@cs.ubc.ca.}
\and
Katrina Ligett\thanks{Hebrew University of Jerusalem. Email: katrina.ligett@mail.huji.ac.il.}
\and 
Adam Wierman\thanks{California Institute of Technology. Email: adamw@caltech.edu.}
\and
Juba Ziani\thanks{University of Pennsylvania. Email: jziani@seas.upenn.edu.}
}

\begin{document}

\maketitle

\begin{abstract}
Motivated by the growing prominence of third-party data providers in online marketplaces, this paper studies the impact of the presence of third-party data providers on mechanism design.  When no data provider is present, it has been shown that simple mechanisms are ``good enough'' -- they can achieve a constant fraction of the revenue of optimal mechanisms. The results in this paper demonstrate that this is no longer true in the presence of a third-party data provider who can provide the bidder with a signal that is correlated with the item type. Specifically, even with a single seller, a single bidder, and a single item of uncertain type for sale, the strategies of pricing each item-type separately (the analog of item pricing for multi-item auctions) and bundling all item-types under a single price (the analog of grand bundling) can both simultaneously be a logarithmic factor worse than the optimal revenue.  Further, in the presence of a data provider, item-type partitioning mechanisms---a more general class of mechanisms which divide item-types into disjoint groups and offer prices for each group---still cannot achieve within a $\log \log$ factor of the optimal revenue.  Thus, our results highlight that the presence of a data-provider forces the use of more complicated mechanisms in order to achieve a constant fraction of the optimal revenue.  
\end{abstract}

\section{Introduction}

Information asymmetries are rampant in markets from ad auctions to art auctions, from acquiring a summer home to acquiring a startup.  Naturally, whenever significant information asymmetries occur, agents have incentives to acquire information through outside channels. As a result, there is a proliferation of companies that seek to collect information that can be sold to participants in auctions with information asymmetries.  Online advertising provides an extreme example. By tracking online behavior, {\em data providers} are able to sell valuable information about internet users (whose attention is the good for sale) to bidders in online advertising auctions. An FTC report \citeyear{FTC} details the scale and prevalence of such data providers --- generating \$426 million in annual revenue in 2012 and growing considerably in the years since. 

Some of the most elegant results in mechanism design focus on providing simple characterizations of (near-)optimal auctions with information asymmetries.
A particularly beautiful example is \citet{Myerson81}'s characterization of optimal single-item auctions. In the special case of one bidder, \citeauthor{Myerson81} characterizes a monopolist's optimal pricing mechanism for a buyer with a private value drawn from a known distribution.  Settings where the item for sale is endowed with a type,
the realization of which affects the buyer's value but is known only to the seller, are particularly well-motivated generalizations of the setting considered by Myerson, and have received substantial attention in the economics literature \citep[e.g.,][]{lewis94,G04,FJM+12, EFG+14, Smolin17}.  Our paper extends this line of inquiry to settings where third-party information about the item  is revealed to the bidders, outside of the control of the seller. 

Specifically, the goal of this paper is to investigate \emph{the impact of third-party data providers on the revenue of simple mechanisms}.  To do this, we consider a simple market---a single seller, a single bidder, and a single good---and a particular form of information asymmetry---the seller knows the type of the good she is selling, but the bidder has only partial information on the item type; the bidder knows his valuation for each of the $n$ possible item types, but the seller knows only distributional information about the valuations. The key to the model is that, in addition to a prior over the item type, the bidder obtains a signal about the item type from a third party data provider and, while the seller can anticipate the signaling scheme used by the data provider, the seller does not know the realization of signals. This captures, for example, a simple model of an ad auction where a mechanism designer sells an ad auction slot to an advertiser, who has incomplete information about the users targeted by the slots.  The advertiser can get additional, third-party information about the target user(s) via data providers that track, for example, users' cookies.

Our model, though stylized, is already general enough to expose the difficulties created by third-party data providers. In particular, our main results show that \emph{simple mechanisms cannot provide revenue within a constant factor of the revenue provided by an optimal mechanism}.  This result is perhaps surprising in the context of an elegant recent paper, where \citet{DPT16} study optimal auctions in the {\em absence} of a data provider. \citet{DPT16}~study the design of simple mechanisms in a setting where the only uncertainty about the item type is that it is drawn from a common prior.  In this context, the question is whether it is valuable for the \emph{seller} to share information with the bidder about the item type, or whether mechanisms that do not reveal information can be approximately optimal. 
Interestingly, \citet{DPT16} are able to characterize optimal auctions for this setting. Their insights show a direct correspondence between mechanisms for selling a single item of uncertain type and multi-item auctions; in particular, this correspondence implies that the seller does not need to reveal any information about the item type to the bidder in order to maximize his revenue. Further combining this observation with the work of~\citet{BILW14} allows them to observe that the better mechanism of two simple approaches---setting a fixed price for the item (the parallel notion of grand bundling, which we term ``item-type bundling''), or pricing each item type separately (the parallel notion of item pricing, which we term ``item-type pricing'')---is guaranteed to yield a constant (in fact, $1/6$) fraction of the revenue of the optimal mechanism.  Thus, in the case where there is no third-party data provider, simple mechanisms are sufficient. 

Our results show, in contrast, that the presence of a third-party data provider, who reveals information outside of the control of the seller, complicates the mechanism design task dramatically.  We first consider revenue-optimal mechanisms. While \citet{DPT16}'s characterization of optimal auctions extends naturally to our setting, these optimal mechanisms may be quite complex.  Concretely, our setting satisfies a type of revelation principle (Lemma \ref{LEM:REVPR}) 
stating that optimal mechanisms require only a single round of bidding, followed by a single round of information revelation (full revelation, in fact); however, such a mechanism presents the bidder with a menu of options that includes an option for each possible valuation vector, combined with each possible posterior of the bidder on the item type after receiving the data provider's signal, and requires the seller to condition the price charged on the realization of the item type. 

Not only does the presence of a data provider complicate the design of the optimal mechanism, our main results show that it also impacts the revenue achievable via simple mechanisms. Specifically, in the presence of a data provider, the better of item-type bundling and item-type pricing may achieve only an $\Omega\left(\log n\right)$ factor approximation of the revenue the seller could have achieved had she
offered a richer menu to the bidder (Theorem~\ref{THM:SIMPLE}); this factor is, in fact, tight (Corollary~\ref{cor:itempricing_log}). In particular, a mechanism that divides the item types into disjoint groups and offers a price on each group can outperform both item-type pricing and item-type bundling by a logarithmic factor.  Such mechanisms are known in the multi-item auction literature as partition mechanisms, and are seen as relatively simple mechanisms \citep[see, e.g.,][]{Rubinstein16}.  In our setting, we refer to such mechanisms as \emph{item-type partition mechanisms}.

This separation between the revenue of item-type partitioning and that of item-type pricing and item-type bundling raises a natural question: \emph{if we expand our view of what constitutes ``simple'' mechanisms to include item-type partitioning, which generalizes both item-type pricing and item-type bundling, can we guarantee that simple mechanisms obtain a constant approximation of the optimal revenue in the presence of a data provider?}  

Our main result uses a more intricate argument to show that this is not the case.  We demonstrate that, in the presence of a data provider, optimal mechanisms can outperform the best item-type partition mechanism by an $\Omega(\log \log n)$ factor (Theorem \ref{THM:ITEM-TYPE-PARTITION}). So, in the presence of a data provider, simple mechanisms truly are no longer optimal. Additionally, our result highlights that the presence of third-party information can simultaneously hurt the optimal revenue achievable by the seller (by a $O\left(\log \log n/\log n\right)$ factor, see Theorem \ref{THM:ITEM-TYPE-PARTITION}).

Our results imply that, in settings where bidders have incomplete information (e.g., ad auctions), it is crucial for a seller whose goal is to maximize revenue via a simple mechanism to have a monopoly on the information available about the good for sale.  A seller may lose significant revenue if using a simple design without a monopoly on data, when the number of possible item types grows large. This can be seen, in particular, in our ad auction example. The value an advertiser gets from showing an ad to a certain individual may depend on many attributes such as gender, age, education, but also specific interests the individual may have. Each unique combination of interests defines a type an individual may have, and there may be exponentially many combinations of such interests. A logarithmic fraction (in the number of types) of the optimal revenue becomes a linear fraction of the optimal revenue in the number of possible interests. Such a revenue guarantee becomes trivial as the number of possible interests grows.

Our discussion so far has focused on the seller, and ignored the data provider's incentives. The results described above do not depend on a specific model of the data provider behavior.  However, when interpreting the lower bounds, it is interesting to consider how the data provider may behave.  Two  particular cases of interest are: (i) a \textit{strategic} data provider that seeks to maximize its profit and (ii) an \textit{adversarial} data provider that seeks to minimize the seller's profit.  

We study the case of a \textit{strategic} data provider in Section~\ref{sec:dp-incentive}, and obtain results about the equilibrium outcomes when the seller and the data provider interact strategically. Specifically, we consider a game played between a seller and a data provider. The game has the seller and the data provider each choosing an action simultaneously. The seller proposes a mechanism that the buyer will engage in, a mechanism which depends on the signaling scheme that the seller expects the provider to offer. The data provider chooses a signaling scheme that it offers to the buyer. Both agents, the seller and data provider, seek to maximize profits.  Our results highlight that, regardless of the form of the mechanism used by the seller, the strategic data provider chooses to reveal all the information that is available to him (Lemma~\ref{LEM:FULL_REVELATION_STRATEGIC}). Importantly, all of the constructions used to prove the lower bounds for simple mechanisms apply to the specific case of strategic data providers, and thus the lower bounds discussed above hold in the presence of a strategic data provider. 

Finally, we consider the case of an \emph{adversarial} data provider in Section~\ref{sec:adversarial}. In this case, the data provider seeks to minimize the seller's profits, which could be the goal if the data provider were also running a platform that competed with the seller.  As in the case of the strategic data provider, our lower bounds can be extended to this setting, and Corollary~\ref{cor: revenue_loss_strategic} highlights that an adversarial data provider can force a revenue of at most $O \left(1/\log n\right)$ of the achievable revenue when no data provider is present. Additionally, this setting is of particular interest because it demonstrates behavior that is, perhaps, counter-intuitive.  Specifically, in contrast to the case of a strategic data provider, a data provider that is attempting to negatively impact the revenue of the seller \emph{may not} want to fully reveal his information about the item type to the bidder (Lemma \ref{lem:adversarial_partial_revelation}).  Instead, there may be intermediate signals which, upon revelation, minimize the revenue of the seller. This serves to highlight the complexity of mechanism design in the context of a third-party data provider, motivating the importance of designing mechanisms that have strong lower bounds regardless of the behavior of the data provider.  

To summarize, in this paper we make the following contributions. We propose a simple model of an auction in the presence of a third-party data provider, capturing information asymmetry regarding the type of the item for sale. Within this model, we first (Section~\ref{sec:charac}) provide a characterization of the optimal auction based on that of~\citet{DPT16}, which may require a complex menu of options. Our main results study the potential for simple mechanisms to approximate the revenue of optimal mechanisms. In Section~\ref{sec:example1}, we show that the item-type equivalents of item pricing and grand bundling cannot achieve within an $\Omega \left( \log n \right)$ factor of the revenue achievable by the optimal mechanism, nor even of the best item-type partition mechanism; this bound is tight, as item pricing can be shown to always achieve a $O \left(\log n \right)$ factor of the optimal revenue (see Corollary~\ref{cor:itempricing_log}). Further, in Section~\ref{sec:example2}, we show that there may be an $\Omega \left( \log \log n\right)$ gap between the revenue of the best item-type partitioning and that of the optimal mechanism. These results highlight that the presence of a data provider significantly reduces the ability of simple mechanisms to approximate the revenue of optimal mechanisms, even in the case of a single seller and a single bidder.  Finally, in Section~\ref{sec:dpmodeling}, we turn to understanding the impact of the behavior of the data provider.  We show that our lower bounds also hold for the specific cases of strategic and adversarial data providers.  Additionally, we show a contrast between these two cases: strategic, revenue-maximizing data providers always fully reveal the information available to them, while adversarial data providers may only {\em partially} reveal information. Thus, \emph{partial} revelation may be more damaging to the seller than {\em full} information revelation.  Taken together, the results in this paper highlight that there is significant motivation for a seller, both in terms of design simplicity and revenue, to be a data monopolist.

\paragraph{Related work} There is a large and rich literature on information and signaling in auctions.  One line of research focuses on designing a signaling scheme (on the part of the seller) given a certain auction format such as the second price auction \citep[see, e.g.,][]{lewis94,G04,ES07,EFG+14, MS12, CCD+15, DIOT15, Smolin17}; another line, closer to our setting, studies the design of \emph{both} the auction and the signaling scheme.  In this line of work there is no data provider; bidders have a prior on their valuation for the item, and any signal on this valuation comes from the seller. More closely related to the current paper, \citet{FJM+12} showed that, if the auctioneer commits to a signaling scheme before choosing the form of the auction, full revelation followed by \citeauthor{Myerson81}'s auction for the revealed item type is the optimal design.
However, \citet{DPT16} revealed the subtlety of this order of commitment and showed that, when the design of the auction and that of the signaling scheme are considered together (without having to commit to one before the other), the optimal strategy is to reveal no information at all, and the overall problem is in fact equivalent to the design of a multi-item auction.  In particular, they show that, when the bidders have a publicly known common prior on the type of the item, the optimal revenue for the seller is that of a multi-item auction.

Furthermore, Theorem $2$ of~\citet{DPT16} shows a one-to-one correspondence between types when selling a single item of uncertain type  and items in a classical multi-item auctions.  In particular, item-type pricing, i.e., mechanisms in which the seller first reveals the item type and then charges a take-it-or-leave-it price, is equivalent to selling separately (i.e., item pricing) in the corresponding multi-item auction, and item-type bundling, i.e., mechanisms in which the seller does not reveal any information and offers a single take-it-or-leave-it price, is equivalent to grand bundling in the corresponding multi-item auction.  When there is a single bidder, \citet{DPT16} further combine this correspondence with results of~\citet{BILW14} to show that the better of item-type pricing and item-type bundling gives at least $1/6$ of the optimal revenue.

The results described above highlight the connection between our work and the study of simple mechanisms for multi-item auctions.  \citet{HN12} pioneered this area. They showed that a seller, using item pricing, can extract a $\Omega\left(1/\log^2 n\right)$ fraction of the optimal revenue from an additive bidder whose values for $n$ items are drawn independently, and selling these items as a bundle can achieve a $\Omega\left(1/\log n\right)$-fraction of the optimal revenue if the bidder's values are i.i.d.  \citet{LiY13} improved the approximation ratio for item pricing to $O\left(1/\log n\right)$, which is tight. \citet{BILW14} showed that, surprisingly, the better of selling separately and grand bundling can achieve at least $1/6$ of the optimal revenue. Subsequently there has been a surge of results generalizing the results of \citeauthor{BILW14}~to broader settings~\citep{CaiH13, Yao15,RubinsteinW15,CaiDW16,ChawlaM16,CaiZ17}. At this point, it is known that simple mechanisms such as sequential two-part tariffs can obtain a constant fraction of the optimal revenue for multiple bidders with combinatorial valuations that are, e.g., submodular, XOS~\citep{CaiZ17}. One might hope to extend these simple deterministic mechanisms to settings where the bidder has correlated values over the items; however, this is impossible.  \citet{HN13} showed that even for a single additive bidder, when valuations are interdependent, the ratio between the revenue obtainable by a randomized mechanism and that of the best deterministic mechanism can be unbounded.

\section{Model and preliminaries}
\label{sec:prelim}
We study mechanism design in settings where buyers may not fully understand how they value the item(s) for sale.  Our guiding example is an ad auction, where the buyer is an advertiser, the seller offers a slot for placing an ad, and the item is a specific user who would view the ad-slot that the buyer is bidding on. The buyer may acquire information about the item from a third-party data provider, for instance via cookie matching. Our paper studies the impact of such third-party information on the mechanism design. 

Even in simple settings, revenue-maximizing mechanisms can be complex: they can involve randomization, and typically let buyers choose among one of (sometimes infinitely) many lotteries~\cite{HN13,DDT14}. This complexity is a barrier for practical implementation. The results of~\citet{DPT16} imply that this complexity carries through to the settings considered in this paper, in which buyers may not know the types of the items they bid for: optimal auctions may be complex and hard to implement in practice even when there is only a single buyer and a single item for sale, that can be of two possible types.

In light of the complexity of revenue-maximizing mechanisms, the literature has focused on simple mechanisms, and on whether they can approximate the optimal achievable revenue; see for example the works of~\citet{HN12,LiY13,BILW14}. In our paper, we focus on analyzing the revenue achievable by such simple mechanisms in the presence of a third-party data provider, and we show an impossibility result: simple mechanisms cannot always achieve a constant fraction (in the number of possible item types) of the optimal revenue. This contrasts with the results of~\citet{BILW14} and~\citet{DPT16}, who show that in the absence of data provider, these simple mechanisms give a constant (in the number of item types) approximation to the optimal revenue.

Our message is that something is impossible, so our results become stronger the simpler the context in which they are proven. As such, we focus on a very simple setting:  a single buyer and a single item for sale. Formally, we consider the following setting. There is a single, revenue-maximizing seller selling a single item to a single buyer. The item for sale takes one of $n$ possible types, and the buyer's valuation may depend on the item type. The buyer does not know the type $\itemtype$ of the item, but has a publicly-known prior $\vec{\pi}$ over the item types. We let $\pi(\itemtype)$ denote the prior probability that the item is of type~$\itemtype$. In the case of ad-auctions, the seller is an ad-platform, the buyer is an advertiser, and the item is a user who will be shown an ad. Item types reflect the heterogeneity in the different types of users, and it is clear that the value to an advertiser of an ad depends on the type of user who will be the ad.

The buyer's private value when the good is of type $\itemtype$ is drawn independently from a publicly known distribution $D(\itemtype)$ over the space of non-negative real numbers $\reals^+$. We denote by $V(\itemtype)$ the buyer's valuation for an item of type $\itemtype$, and  by $\vec{V} = \left(V(1), \ldots, V(n) \right)$ the buyer's valuation vector. 

In our setting, there is a third-party data provider who has (potentially imperfect) information on the type of the item, in the form of a random variable~$X$ that can be arbitrarily correlated with the type of the item. This is unlike the setting of~\cite{DPT16} in which only the seller could reveal information about the item type to the bidders.
The joint distribution of $\itemtype$ and $X$ is publicly known, but the realized value of~$X$ is only observed by the data provider. 
	
	The data provider designs a \emph{signaling scheme} in the form of a function $S$ that maps $X$ to $\Delta(\Alphabet)$, the set of distributions over an alphabet~$\Alphabet$.  The data provider is able to commit to such a scheme and, on observing information~$X$, the data provider draws a signal $\sig$ from~$\Alphabet$ according to the distribution~$S(X)$, and sends it to the buyer if the latter purchases from the data provider.\footnote{For most of the paper we make no assumptions about the data provider; however, in Section \ref{sec:dpmodeling}, we take the incentives of the data provider into account and consider two different scenarios: one in which the data provider is strategic and aims to maximize his revenue from selling his signal, and one in which he is adversarial and tries to hurt the seller's revenue.}
	
	After receiving $\sig$, the buyer updates his prior using Bayes' rule. We denote the resulting posterior by $\vec{\pi}_{\sig} \in \Delta([n])$. 
	The buyer aims to maximize his utility given his posterior on the item type; we assume utilities are quasi-linear. Since the realization of $X$ was not visible to the seller, if the buyer purchases from the data provider, the seller would know only the distribution over the buyer's posteriors, conditioning on the item type~$\itemtype$.

	Motivated by the sale of online advertisements, where the sale repeats rapidly with the item type redrawn in each round, we assume that the buyer's decision to purchase from the data provider is made before his value is realized. In this setting, the buyer, seller and data provider act as follows:
\begin{enumerate}
\item The seller commits to a information revelation policy and a mechanism. Simultaneously, the data provider commits to a signaling scheme.
\item The buyer decides whether to enter a contract with the data provider and to purchase his information.
\item All the participants receive their private information: the buyer observes his valuation vector, the data provider observes $X$, and the seller observes the item type. 
\item The buyer sees the signal from the data provider. The seller reveals information to the buyer and runs her mechanism. 
\end{enumerate}
The roles of seller and data provider are asymmetric,  in that the buyer makes the decision of purchasing the data provider's signal before his valuation vector is realized, whereas the purchase decision with the seller is made after the buyer realizes his valuation vector.  In the ad auction setting, this asymmetry captures the practice that data sets are often sold in batch, or as ``right of access,'' whereas ads are sold per impression, often via bidding in an auction for each individual ad.  This asymmetry is even more marked when the buyer is an agency that bids on behalf of many advertisers in individual auctions but buys data in batch to inform such bidding.

\subsection{Simple mechanisms} The goal of this paper is to show that, in the presence of a data provider, no simple mechanism can extract a constant fraction of the optimal revenue.  To provide context for our work, we first summarize results on the existence of simple mechanisms when no data provider is present. 

In the absence of a third-party data provider, in the single seller, single buyer, {\em multi-item} setting, {\em where each item has a single type}, \citet{BILW14} show that, although the optimal mechanism may be complex, a simple mechanism achieves a constant factor of the optimal revenue.  In particular, this mechanism is simply  the better of either item pricing or grand bundling.  This result was originally stated for multi-item auctions, but the results of \citet{DPT16} show that the current setting, with a single item that can take on multiple possible types, in fact reduces to the multi-item auction setting. We provide more details on this reduction in Appendix~\ref{app: reduction}.

These simple mechanisms are important throughout our paper, so we formally define them here, in the context of selling a single item with multiple possible types.

\begin{definition}
\emph{An \textbf{item-type pricing} mechanism first reveals the type $i$ of the item to the buyer, then offers to sell the item to the buyer at some price $P_{\text{it}}(i)$. We  refer to such mechanisms as ``selling the types separately,'' in analogy to the concept of selling separately in  multi-item auctions.}
\end{definition}

\begin{definition}
\emph{An \textbf{item-type bundling} mechanism offers the item for sale at some price $P_{gr}$ without revealing any information about the realized type of the item.} 
\end{definition}

\noindent The following results summarize the power of these simple mechanisms in the single-item, multi-type setting, without a data provider. 

\begin{proposition}[\citet{HN12,LiY13,DPT16}]~\label{prop:log_separate}
In the absence of a data provider, both item-type pricing and item-type bundling yield at least a $\Omega\left(\frac{1}{\log n}\right)$-approximation to the optimal revenue when there is a single seller, a single buyer, a single item for sale, and the buyer has a publicly known prior over the type of the item. 
\end{proposition}

\begin{proposition}[\citet{BILW14,DPT16}]~\label{prop:constant-factor}
In the absence of a data provider, the maximum of item-type pricing and item-type bundling yields at least a $\frac{1}{6}$-approximation to the optimal revenue when there is a single seller, a single buyer, a single item for sale, and the buyer has a publicly known prior over the type of the item. 
\end{proposition}

In addition to the notions of simple mechanisms considered in previous work, in this paper we consider a natural generalization of both item-type bundling and item-type pricing mechanisms. The mechanisms we consider are called item-type partition mechanisms and are simple in the sense that (i) they do not require randomization and (ii) near-optimal such mechanisms can be efficiently computed, as shown by~\citet{Rubinstein16}. The class of partition mechanisms covers many simple mechanisms for single bidder multi-item auctions, such as those considered by~\cite{HN12,LiY13,BILW14,feldman2014combinatorial,bateni15,morgenstern2016learning,rubinstein2018simple,chen2018complexity}.

\begin{definition}\label{def:item-type partitioning} 
\emph{An \textbf{item-type partition} mechanism first partitions the set of item types into non-empty groups $\cG_1$ to $\cG_g$, priced (resp.) $P_1$ to $P_g$. The  mechanism then observes the type $i$ of the item, and offers the item at price $P_r$, where $r$ is uniquely chosen such that $i \in \cG_r$.}
\end{definition}

Note that, after observing the offered price $P_r$, the buyer may infer that the realized item type must belong to group $\cG_r$. Item-type pricing is an instantiation of item-type partitioning where the partition contains a separate group for each type; item-type bundling corresponds to item-type partitioning using the trivial partition. Item-type partitioning is, however, significantly more powerful than these other simple mechanisms, as it allows the seller to partition the item types into arbitrarily many groups of arbitrarily many sizes. 

\subsection{The equal revenue distribution} 

Our lower bounds for simple mechanisms are based on settings when the buyer's values are taken from the so-called equal revenue distribution. 

The equal revenue distribution is a natural benchmark and often used in the literature that seeks to inform practical auction design (see, for example, \cite{dughmi2009revenue}, \cite{hartline2009simple}, \cite{HN12}, \cite{sundararajan2016prediction}, \cite{paes2016field}, and \cite{medina2017revenue}). It belongs to the class of Pareto distributions, which are often used to capture real-life situations in which most of the value is held by a small part of a given population.  Further, among the Pareto distributions, the equal-revenue distribution captures the property of constant revenue and yields constant virtual values (as defined by Myerson~\cite{Myerson81}). Thus the equal-revenue distribution acts as a ``borderline'' element in the space of regular distributions, to which Myerson's optimal auction applies. Finally, the equal-revenue distribution is central to the design of simple and near-optimal mechanisms because it provides a worst-case example of when selling items separately in multi-item auctions may (counter-intuitively) perform significantly worse than selling a single bundle containing all items, even when the buyer's valuations for each item are independent and identically distributed~\cite{HN12}. In our setting, third-party signaling does not affect the revenue from item-type pricing, as the seller must eventually reveal the item type and override the data provider's signal; however, such signaling may significantly impact the revenue obtained from bundling items together. As such, the equal-revenue distribution is crucial to a number of our constructions. It is defined as follows.

\begin{definition}
\emph{A random variable $X$ with support $[1,+\infty)$ follows the \textbf{equal revenue (ER) distribution} if and only if  $\Pb \left[ X \leq x \right] = 1-\frac{1}{x}$.}
\end{definition}

The equal revenue distribution gets its name from the fact that it has constant virtual value, and every price in the distribution's support offers the same expected revenue. The equal revenue distribution also has a number of other useful properties, proved by~\citet{HN12}, which we summarize here. Unless otherwise specified, $\log$ is taken to be the natural logarithm.

\begin{lemma}[\citet{HN12}]\label{lem: ER_property}
Let $n \geq 2$ be an integer, and let $I_1, \ldots, I_n$ be $n$ i.i.d random variables that follow the ER distribution. Then $\Pb \left[\frac{1}{n} \sum_{i=1}^n I_i \geq \frac{\log n}{2} \right] \geq \frac{1}{2}$, and for any $P \geq 6  \log n$, $\Pb \left[\frac{1}{n} \sum_{i=1}^n I_i \geq P \right] \leq \frac{9}{P}$.
\end{lemma}

\section{Optimal mechanisms in the presence of a data provider}\label{sec:charac}

Before focusing on simple mechanisms, we first explore optimal mechanisms (which are not simple).  Specifically, in this section we provide a characterization of optimal mechanisms for a single buyer and a single item, with several possible item-types, in the presence of a third-party data provider who knows (possibly imperfect) information about the item type, and who reveals some of this information to the buyer. Note that the data provider is represented via a signaling scheme that, from the model perspective, is subsumed into a probability distribution over posteriors $\vec{\pi}$, representing beliefs of the buyer regarding the item's type. Therefore, a buyer with access to the data provider's signal has private information in the form of a valuation $\vec{V}$ and a posterior $\vec{\pi}$.

We focus on a class of mechanisms that allow the seller to charge the buyer a price that is conditional on the type of the item. Restricting our attention to such type-contingent price mechanisms is without loss of generality. The characterization we present is a type of revelation principle, similar to that presented in \citet{DPT16}, where the difference is the presence of a data provider.

Our main result in this section shows that there always exists an optimal mechanism that takes the form of a conditional price menu. Thus, before stating our characterization, we need  to formally define a conditional price menu.  
We postpone the proof of Theorem~\ref{LEM:REVPR} and discussion of the characterization to Appendix~\ref{appx:charac}.

\begin{definition}
\emph{A \textbf{menu with conditional prices} is a fixed collection of pairs $(\vec{Z},\vec{P})$, where each $\vec{Z} \in[0,1]^n$ is called an allocation rule, and each $\vec{P} \in \reals^n_+$ is called a pricing rule. The buyer selects at most one pair $(\vec{Z},\vec{P})$. After his choice has been made, the type $i$ of the item is revealed.  Given item type $i$, the buyer pays price $P(i)$, and receives the item with probability $Z(i)$.}
\end{definition}

\begin{theorem}\label{LEM:REVPR}
For any equilibrium of any mechanism $\cM$ in the presence of a data provider, such that the buyer, conditioned on the realization of her valuation vector and posterior beliefs over item types given the signal from the data provider, obtains non-negative payoff in expectation, there is a conditional price menu that is incentive compatible, interim individually rational, and provides the same revenue.
\end{theorem}

Theorem~\ref{LEM:REVPR} implies the optimal revenue is given by the solution of a linear program whose size is proportional to the number of possible pairs of value vectors and posteriors. We make use of this linear program in Section \ref{sec:adversarial}. Additionally, Theorem~\ref{LEM:REVPR} can be extended to the multi-buyer setting, in which case one can write the optimal mechanism as an interim individually rational direct revelation mechanism with no information revelation by the seller required prior to bidding, that is the solution to a linear programming problem. We state this extension formally in Appendix~\ref{appx:charac}.

Unfortunately, while such conditional price menus provide a simple and linear characterization of the optimal revenue, they are often intractable. In particular, the optimal mechanism grows linearly and the linear program that yields it grows quadratically in the number of possible pairs of value vectors and posteriors. This motivates the study of simple and tractable mechanisms, and of whether they achieve near-optimal revenue in the presence of a data provider.

\section{Simple mechanisms in the presence of a data provider}\label{sec:lower_bounds}

The complexity of the optimal mechanism highlighted by the characterization in the previous section motivates the task of designing simple mechanisms that can achieve a constant fraction of the optimal revenue.  Our results in this section highlight that such a goal is impossible.  In particular, our results give upper bounds on the fraction of the revenue simple mechanisms can be guaranteed to achieve for two notions of ``simple.''  Specifically, in Section~\ref{sec:example1}, we focus on a class of ``simple'' mechanisms defined as the better of item-type pricing and item-type bundling, echoing~\citet{BILW14} and in Section~\ref{sec:example2}, we extend our discussion of ``simple'' mechanisms to include item-type partition mechanisms.

Before moving to our results, it is important to note that the works of~\citet{HN12} and~\citet{LiY13} provide lower bounds on the revenue achievable by simple mechanisms.  An immediate consequence of Proposition~\ref{prop:log_separate} is the following\footnote{This corollary is a direct consequence of the fact the revenue from item-type pricing is not affected by the presence of a data provider, since item-type pricing requires the seller to override the provider's signal by fully revealing the item type to the buyer.}:

\begin{corollary}\label{cor:itempricing_log}
In the presence of a data provider, the better of item-type pricing and item-type bundling yields at least a $\Omega\left(\frac{1}{\log n}\right)$-approximation to the optimal revenue when there is a single seller, a single buyer, a single item for sale, and the buyer has a publicly known prior over the type of the item. 
\end{corollary}

This corollary provides context for the results that follow, which show that the revenue from item-type bundling can be negatively impacted by the presence of a data provider. In particular, a seller may not be able to obtain more than a $O\left(\frac{1}{\log n}\right)$-approximation to the optimal revenue via the better of item-type pricing and item-type bundling. Thus, the bound we provide on such mechanisms is tight.

\subsection{Item-type pricing and item-type bundling}\label{sec:example1}

Our main results focus on bounding the revenue achievable via simple mechanisms in the presence of a third party data provider. In this section, we focus on a class of simple mechanisms in which the seller runs the better of item-type pricing and item-type bundling.  These are particularly interesting mechanisms to consider given Proposition~\ref{prop:constant-factor}, where  \citet{DPT16}, using results of \citet{BILW14}, show that this style of mechanism obtains a constant fraction of the optimal revenue when a data provider is not present. In contrast, we show here that, in the presence of a data provider, the better of item-type pricing and item-type bundling cannot always achieve a constant fraction, and instead may only achieve an $O \left(\frac{1}{\log n}\right)$ fraction of the optimal revenue. As per Corollary~\ref{cor:itempricing_log}, this fraction is tight.

\begin{theorem}\label{THM:SIMPLE}
There exists a single seller, single bidder, single  item (taking one of $n$ item types) setting where, in the presence of a data provider who signals information about the item type realization to the bidder, the expected revenue of the better of item-type pricing and item-type bundling is no more than a $O\left(\frac{1}{\log n}\right)$ fraction of the expected revenue of the optimal mechanism. More specifically:
\begin{itemize}
\item In the absence of a data provider, the optimal revenue of the seller is $\Theta \left( \frac{\log^2 n}{\sqrt{n}} \right)$. The optimal revenue for item-type pricing is $\Theta \left( \frac{\log n}{\sqrt{n}} \right)$, and the optimal revenue from item-type bundling is $\Theta \left( \frac{\log^2 n}{\sqrt{n}} \right)$.
\item In the presence of a data provider, the optimal revenue is $\Theta \left( \frac{\log^2 n}{\sqrt{n}}\right)$. Both the optimal revenue from item-type pricing and the optimal revenue from item-type bundling are $O \left( \frac{\log n}{\sqrt{n}}\right)$.
\end{itemize}
\end{theorem}

This theorem highlights that there exists a setting where the introduction of a data provider does not affect the optimal revenue, but where the data provider's presence is harmful to the optimal revenue of the better of item-type pricing and item-type bundling. 

The proof of Theorem~\ref{THM:SIMPLE} relies on constructing a family of instances of our problem for which there is a gap between the optimal revenue and the revenue of simple mechanisms. The construction is given below:

\begin{construction} \label{ex:simple}
Let $n=m^2$ be the number of item types, for some integer $m$. The types are partitioned into $m$ groups $I_1,\ldots,I_m$ such that each group contains exactly $m$ types. The bidder's prior on the item type is uniform, i.e., the bidder initially believes that each item type is realized with probability $1/n=1/m^2$, and that the probability that the realized type belongs to group $I_k$ is therefore $1/m$. The bidder's valuation for type $i$ in group $I_k$ is $V(i)/k$, where $V(i)$ is a random variable drawn from the equal revenue distribution. The bidder's valuations for different item types are drawn independently of each other.

In this setting, we allow the data provider to observe to which group the item type belongs. The data provider fully reveals this information to the bidder. We show later (Section \ref{sec:dpmodeling}) that this is the signaling scheme that a strategic, revenue-maximizing data provider would sell in this scenario, and that the buyer would always opt to buy the data provider's signaling scheme; therefore, our results extend to the case of a strategic provider.

Given the data provider's signal, the bidder's posterior probability on the item being of type $i$, upon observing signal $s_k$ informing him that the group is $I_k$, is given by
$
\pi_{s_k}(i) = 
\begin{cases} 
	0 & i \notin I_k \\
   	\frac{1}{m} & i \in I_k\\
\end{cases}.
$
\end{construction}

Given this construction, the proof seeks to characterize the revenue of the optimal and simple mechanisms for such instances.  We highlight the structure of the proof by organizing it as the following sequence of lemmas, which we prove in Appendix~\ref{app:simple}.  We emphasize that, while item-type pricing is unaffected by the introduction of a data provider, the presence of a data provider can harm the optimal revenue of other classes of mechanisms.

\begin{lemma}\label{lem:item_simple}
The expected revenue from optimal item-type pricing in Construction~\ref{ex:simple} is $\Theta \left( \frac{\log n}{\sqrt{n}} \right)$, independently of whether a data provider is present.
\end{lemma}

\begin{lemma}\label{lem:bundling_simple_noDP}
The expected revenue from optimal item-type bundling in Construction~\ref{ex:simple} is $\Theta \left( \frac{\log^2 n}{\sqrt{n}} \right)$ in the absence of a data provider.
\end{lemma}

\begin{lemma}\label{lem:bundling_simple}
The expected revenue from item-type bundling in Construction~\ref{ex:simple} is $O \left( \frac{\log n}{\sqrt{n}} \right)$ in the presence of a data provider.
\end{lemma}

\begin{lemma}\label{lem:rev_opt}
There exists an item-type partition mechanism that achieves expected revenue $\Omega \left( \frac{\log^2 n}{\sqrt{n}}\right)$ in Construction~\ref{ex:simple}. The optimal revenue in Construction~\ref{ex:simple} is $\Theta \left( \frac{\log^2 n}{\sqrt{n}}\right)$.
\end{lemma}

\subsection{Item-type partitioning}\label{sec:example2}

The previous section shows that neither item-type pricing nor item-type bundling, nor the better of the two, can always achieve a constant fraction of the optimal revenue in the presence of a data provider.  However, one may wonder if the result is due to the restrictive nature of the ``simple'' mechanisms considered.  Here, we show that, in the presence of a data provider, even the more general class of item-type partition mechanisms is insufficient to guarantee a constant fraction of the optimal revenue. This is particularly tantalizing due to the fact that Construction~\ref{ex:simple} admits an item-type partition mechanism that yields a constant approximation to the optimal revenue. However, in this section, we show a construction where the best item-type partition mechanism only achieves a $O\left(1/\log \log n\right)$ fraction of the optimal revenue.  Note that this implies that our ``simpler'' simple mechanisms, item-type bundling and item-type pricing, also do not yield a constant fraction of the optimal revenue, since they are special cases of item-type partitioning. 

\begin{theorem}\label{THM:ITEM-TYPE-PARTITION}
There exists a single seller, single bidder, single item (taking one of $n$ item types) setting where, in the presence of a data provider who signals information about the item type realization to the bidder, no item-type partition mechanism can achieve revenue higher than $O\left(\frac{1}{\log \log n}\right)$ of the optimal revenue. More specifically:
\begin{itemize}
\item In the absence of a data provider, the optimal revenue of the seller is $\Theta \left(\log n \right)$. The optimal revenue from item-type partitioning is $\Theta(\log n)$, and is achieved for the item-type bundling partition. 
\item In the presence of a data provider who signals information about the item type realization to the bidder, the optimal revenue is $\Theta \left( \log \log n \right)$, and is achieved by Mechanism~\ref{mech:loglog} below. The optimal revenue from item-type partition mechanisms is $\Theta(1)$.
\end{itemize}
\end{theorem}
This theorem illustrates a setting where the introduction of a data provider decreases the revenue of optimal mechanisms, and where restricting to item-type partitioning mechanisms further harms revenue in the presence of a data provider. Therefore, there could be strong incentives for a seller to be a data monopolist, particularly if the seller has a preference to run simple mechanisms. 

Again, the proof of the theorem relies on a construction, which we detail in the following.

\begin{construction}\label{ex:conditioning}
Given an integer $m$, let $n=2^{m}$ be the number of item types. 
The bidder's prior on the item type is uniform, i.e., the bidder initially believes the item type takes each $i \in [n]$ with probability $1/n$. 
The bidder's valuation for each type is drawn i.i.d. from an equal revenue distribution.

We consider $m$ possible partitions of the $n$ types. Given a particular $k \in [m]$, we partition the set of all types into $n_k=2^{m-k}$ subsets of size $2^{k} \geq 2$ each. Specifically, for $k\in [m]$ we partition the set of types into the subsets $I_{k,1}$ to $I_{k,n_k}$, where $I_{k,j} = \{(j-1) \cdot 2^k + 1, \ldots, j \cdot 2^k\}$ for all $j \in [n_k]$.

The information we allow the data provider to observe is structured as follows. First, a value of $k \in [m]$ is drawn  according to the following distribution: for $k \leq m-1$, $k$ is drawn with probability $\frac{1}{k(k+1)}$; $m$ is drawn with the remaining probability $\frac{1}{m}$. The value $k$ is drawn, importantly, independently of the type $i$ of the item. Then, the data provider observes which group of size $n_k$ (i.e., among $I_{k,1}$ to $I_{k,n_k}$) the item type lies in. Given the observations, the data provider reveals his full information to the bidder, namely, exactly which group of size $n_k$ the item type belongs to.  We show later (Section \ref{sec:dpmodeling}) that this is what a strategic, revenue-maximizing data provider would reveal in this scenario. We denote by $s_{k,j}$ the realization of the signal that indicates to the bidder that the item belongs to group $I_{k,j}$. We call $k$ the size indicator.
\end{construction}

We break the proof of the theorem down into a sequence of lemmas, each of which is proven in Appendix~\ref{app:item-type-partition}.  Note that the first part of the theorem, when a data provider is not present, is a direct consequence of~\citet{HN12} and~\citet{BILW14}. 

\begin{lemma}~\label{lemma:partition}
The expected revenue from the optimal item-type partition mechanism is $O(1)$ in Construction~\ref{ex:conditioning}.
\end{lemma}

\begin{lemma}~\label{lemma:loglog}
There exists a mechanism that yields revenue at least $\Omega \left( \log \log n \right)$ in Construction~\ref{ex:conditioning}. The optimal revenue in Construction~\ref{ex:conditioning} is $\Theta \left( \log \log n \right)$.
\end{lemma}

To prove Lemma~\ref{lemma:loglog}, we first construct a mechanism that achieves revenue $\Omega(\log\log n)$ in Construction~\ref{ex:conditioning}.  In particular, we consider the following design.
\begin{mechanism}\label{mech:loglog}
The seller offers a menu of $\sum_{k=1}^{m} n_k$ options. For every $\kappa \in [m]$, and every $ \iota \in [n_\kappa]$, the menu contains the following option $L_{\kappa,\iota}$: the bidder first pays $P_{\kappa} = \frac{1}{8} \log 2^{\kappa} = \frac{ \log 2}{8} \kappa$, then gets the item if and only if it is in group $I_{\kappa,\iota}$. Note that the price only depends on $\kappa$.  
\end{mechanism}
To show that Mechanism~\ref{mech:loglog} 
 yields revenue $\Omega \left( \log \log n \right)$ in Construction~\ref{ex:conditioning},
we need the following lemma, which characterizes the bidder's behavior in the mechanism. More specifically, we show in the following lemma that if the bidder receives signal $s_{k,j}$, he purchases the corresponding option $L_{k,j}$ in Mechanism~\ref{mech:loglog} with probability almost $1$.

\begin{lemma}\label{clm: inter_ex2} In the the setting of  Construction~\ref{ex:conditioning}, suppose the bidder receives signal $s_{k,j}$ (indicating that the item belongs to group $I_{k,j}$ of size $2^{k}$) for $k \geq 2 \cdot 10^2 + 1$. Consider the menu of options proposed by the seller in Mechanism~\ref{mech:loglog}. 
 With probability at least $1-10^{-3}$, no option $L_{\kappa,\iota}$ with either $\kappa \ne k$ or $\iota \ne j$ yields a higher utility for the bidder than option $L_{k,j}$, and $L_{k,j}$ yields positive utility to the bidder.
\end{lemma}

We remark that Mechanism~\ref{mech:loglog}, although it has a concise description, is not ``simple'' in any of the usual senses, and is in fact carefully tailored to the incentives of the bidder. We do not know of ``simpler'' mechanisms that are approximately optimal in this setting.

\section{The behavior of the data provider}
\label{sec:dpmodeling}

So far, we have not discussed the behavior of the data provider.   The characterization of  optimal mechanisms, and our bounds on the achievable revenue of simple mechanisms in the previous sections, do not depend on specific assumptions about the behavior of the data provider. However, it is useful to consider specific models of the data provider when interpreting our lower bounds.  In particular, one may wonder if more positive results are possible for some standard game-theoretic models of the interaction between the data provider and the seller.

To this end, it is natural to consider two extreme models for a data provider's incentives: (i) the data provider may be \textit{strategic}, seeking to maximize his revenue from selling his information, or (ii) the data provider may be \textit{adversarial}, seeking to minimize the profits of the seller. In this section, we characterize the behavior of the data provider in each of these models. Our results provide a clear contrast between the two models: strategic data providers always reveal all of their information to the bidder, while adversarial data providers may not reveal all available information to the bidder. This has implications for mechanism design, highlighting the importance of mechanisms with good worst-case bounds on revenue, independent of data provider behavior.

Importantly, the proofs of the bounds on the achievable revenue of simple mechanisms in Theorems~\ref{THM:SIMPLE} and~\ref{THM:ITEM-TYPE-PARTITION} use constructions that can be interpreted as complete revelation by the data provider.  Thus, those bounds apply to the case of a strategic data provider.  Concretely, this means that, in the presence of a strategic, revenue-maximizing data provider, simple mechanisms cannot guarantee near-optimal revenue for the seller.  

\subsection{A strategic data provider}
\label{sec:dp-incentive}

One may hope that the impossibility results of Section~\ref{sec:lower_bounds} no longer apply if the data provider is strategic, and seeks to reveal data to maximize its revenue.  However, in this section, we show that the results directly extend to the case of a revenue-maximizing data provider. The reason for this the following lemma, which highlights that the data providers revenue-maximizing scheme is always to fully reveal its information. In particular, all the examples we use to prove Theorems~\ref{THM:SIMPLE} and~\ref{THM:ITEM-TYPE-PARTITION} make use of constructions where the data provider uses full revelation.  Thus, the implications of those theorems also hold under the assumption of a strategic data provider. 
 
\begin{theorem}\label{lem:opt_fullinfo}
The fully-revealing signaling scheme is revenue-maximizing for the data provider.  
\end{theorem}

\begin{proof}
Formally, let $\Alphabet$ be the range of~$X$, and let $S^*(X)$ be the distribution where all probability is point massed on~$X$.  In what follows, we use $S^*$ to denote this fully-revealing signaling scheme. 

We first show that, for any mechanism adopted by the seller, and for any buyer's value vector, a fully revealing signaling scheme maximizes the ex-ante expected utility of the buyer. Let us fix a mechanism $\cM$ and valuation vector $\vec V$.  Let $S$ be an arbitrary signaling scheme that maps the data provider's information~$X$ to $\Delta(\Alphabet)$. Recall that the buyer forms a posterior distribution $\vec \pi_s$ over the item type~$\itemtype$ when receiving signal $\sig$ drawn from $S(X)$.  We denote by $U(\vec V, S)$ the buyer's utility in~$\cM$ when her value vector is~$\vec V$ and he purchases signaling scheme~$S$ from the data provider.  

A key technical lemma for our argument is the following, which is provided in Appendix~\ref{app:strategic}.

\begin{lemma}\label{LEM:FULL_REVELATION_STRATEGIC}
For any valuation vector $\vec V$, any information $X$ received by the data provider, and any signaling scheme $S(.)$, $U(\vec V, S^*) \geq U(\vec V, S)$.
\end{lemma}

Given this lemma, let $U(\vec V)$ be the buyer's utility if his value vector is~$\vec V$ and he does not purchase from the data provider.  The ex-ante value of a signaling scheme~$S$ for the buyer is then $\Ex[\vec V]{U(\vec V, S) - U(\vec V)}$. As in our model, the buyer must decide whether to enter a contract with the data provider and by his signaling scheme before the buyer observers the realization of his own valuation for each type, this difference is exactly the highest price the buyer is willing to pay for the scheme~$S$. Equivalently, this is the maximum price the data provider can charge, and as such his maximum achievable revenue under signaling scheme $S$. Because the difference is maximized when $\Ex[\vec V]{U(\vec V, S)}$ is maximized, which happens when $S = S^*$ by Lemma~\ref{LEM:FULL_REVELATION_STRATEGIC}, we directly obtain that $S^*$ is revenue maximizing for the data provider.
\end{proof}

Building on the previous theorem, the following corollary highlights the contrast in revenue between the settings where a data provider is or is not present; highlighting the damaging impact of a third-party data provider for the seller. (Recall that item-type partition mechanisms yield at least $\frac{1}{6}$ of the optimal revenue when no data provider is present \citep{BILW14}.) 

\begin{corollary}\label{cor: revenue_loss_strategic}
There exists a single seller, single bidder, single item setting where, in the presence of a revenue-maximizing data provider, no item-type partition mechanism can achieve revenue higher than $O\left(\frac{1}{\log n}\right)$ of the optimal revenue achievable by an item-type partition mechanism when no data provider is present. More precisely:
\begin{itemize}
\item In the absence of a data provider, the optimal revenue for the seller is $\Theta\left(\log n\right)$, and is attained by item-type bundling.
\item In the presence of a data provider, the optimal revenue is $1$, and is attained by item-type pricing.
\end{itemize}
\end{corollary}

\begin{proof}
Consider a setting with $n$ item types, distributed i.i.d., according to an Equal Revenue distribution. The optimal revenue in the absence of a data provider is $\Theta (\log n)$ by~\cite{HN12}. As 
\[
P\cdot\Pr \left[ V(i) \geq P \right] = P \cdot \frac{1}{P} = 1
\]
for all $P$ when $V(i)$ follows an equal revenue distribution, the optimal revenue in the presence of a revenue-maximizing data provider that exactly knows and reveals the item type is $1$. On the other hand, for independent valuations, the seller can always guarantee a $\Omega \left( \frac{1}{\log n}\right)$ of the revenue by revealing the item type and pricing optimally, as per~\cite{LiY13}. Note that this example is a special case of Construction~\ref{ex:conditioning}.
\end{proof}

\subsection{An adversarial data provider}~\label{sec:adversarial}

We now consider an \emph{adversarial} data provider, who aims
to minimize the seller's revenue.  In contrast to what happens in the case of a strategic data provider, the main result in this section shows that revealing {\em less} information can sometimes be \textit{more} damaging to the seller's revenue. This phenomenon, however, does not occur when the data provider has perfect information about the item type. 

Recall from Section~\ref{sec:dp-incentive} that we use $S^*$ to denote the signaling scheme that reveals full information, and let us denote $REV(S)$ to denote the seller's optimal revenue when the data provider adopts the signaling scheme~$S$. The theorem below formally states our counter-intuitive result that the data provider can minimize the revenue of the seller by only {\em partially} revealing information.

\begin{theorem}\label{lem:adversarial_partial_revelation}
Let the number of item types be $n=2$.
There exists a distribution over the buyer's valuations $\vec{V}$, a prior $\vec{\pi}$ over the item type and a partially informative distribution over the data provider's information~$X$, such that there is a signaling scheme $S$, with $REV(S) < REV(S^*)$.
\end{theorem}

The proof of this result uses the following construction:
\begin{construction}~\label{ex:nonmonotone}
Let the bidder's valuation for each item type be drawn i.i.d., taking value $1$ with probability $1/2$ and value $2.1$ with probability $1/2$. The bidder and the data provider share a common prior $\vec{\pi} =(3/4,1/4)$. That is, they both initially believe the item is of type $1$ with probability $3/4$ and of type $2$ with probability $1/4$. The data provider receives information $X$ on some support $ \{x_1,x_2\}$. If the item type is $1$, the provider receives $x_1$ with probability $2/3$ and $x_2$ with probability $1/3$, and if the item type is $2$, the provider receives $x_2$ with probability $1$.
\end{construction}

\begin{proof}
We show that in Example~\ref{ex:nonmonotone}, the data provider has a signaling scheme under which the seller's optimal revenue is less than under the fully-revealing scheme.

 If the data provider reveals full information, then with probability $1/2$, the bidder receives $x_1$ and has posterior $\vec{\pi}_{x_1} = (1,0)$ (when receiving $x_1$, the provider knows the item must be of type~$1$); with probability $1/2$, he receives $x_2$ and thus has posterior $\vec{\pi}_{x_2} = (1/2,1/2)$. Remember that by the characterization of Section~\ref{sec:charac}, the seller's best response to the data provider can be written as an interim individually rational mechanism that does not require any information revelation; further, the optimal revenue from such mechanisms can be obtained by solving a linear program. Computing the seller's optimal revenue via linear programming, we see that the revenue is $REV(S^*) = 1.1062$.

Consider the following partially revealing signaling scheme: let $\Sigma$, the range of the signaling scheme, be $\{\partsig_1, \partsig_2, \partsig_3\}$; let $\varphi$ be the mapping $\varphi(x_i) = \partsig_i$ for $i = 1, 2$;  
when the provider receives realization $x$ of $X$ that belongs to $\{x_1,x_2\}$, the provider outputs $\varphi(x)$ w.p. $1-\epsilon = 0.86$ and outputs $\partsig_3$ with probability $\epsilon = 0.14$. Given this signaling scheme, when the bidder receives signal $\partsig_1$ (which occurs with probability $\frac{1}{2}(1-\epsilon) = 0.43$), he infers $X=x_1$, and so his posterior is $\vec{\pi}_{\partsig_1} = \vec{\pi}_{x_1}$. Similarly, when he receives signal $\partsig_2$ (which also occurs with probability $0.43$), the bidder infers that it must be the case that $X=x_2$, hence he has posterior $\vec{\pi}_{\partsig_2} = \vec{\pi}_{x_2}$. Finally, when the bidder receives $\partsig_3$ (which occurs with probability $\epsilon = 0.14$), he infers that $X= x_1$ or $X = x_2$ with equal probability by symmetry, and hence his posterior is $\vec{\pi}_{\partsig_3} = \frac{1}{2} \left(\vec{\pi}_{x_1} + \vec{\pi}_{x_2} \right) = \left(3/4,1/4 \right) = \vec{\pi}$. Computing the optimal revenue of the seller via linear programming, using the the results of Section~\ref{sec:charac}, we get that the revenue is only  $REV(S) = 1.0991 < 1.1062 = REV(S^*)$. 
\end{proof}

It may seem counter-intuitive that the data provider can harm the seller {\em more} by providing {\em less} information to the bidder. After all, one consequence of the characterization of Section~\ref{sec:charac} is that the seller can only lower her revenue by revealing more information to the bidder. However, information from the provider and information from the seller are not equivalent from the perspective of the seller, because the seller does not get to see the realization of the signal that the provider sends to the bidder. When the seller reveals information, she knows exactly what the bidder's posterior is, and can act as a function of the realized posterior; she is then faced with exactly the problem solved by \citet{DPT16} for that realized posterior. When the data provider reveals information, the seller, who only knows the signaling scheme but not the signal, faces a distribution of posteriors and does not know which of them is correct.

In particular, in Example~\ref{ex:nonmonotone}, in the fully-revealing signaling scheme there are two posteriors $\vec{\pi}_{x_1}$ and $\vec{\pi}_{x_2}$. Each
occurs with probability $1/2$. The seller, intuitively, wishes to design a menu with one option for each of these two posteriors.
In the partially revealing signaling scheme there is a third posterior, $\vec{\pi}_{\partsig_3} = \vec{\pi}$, which is an average
of $\vec{\pi}_{x_1}$ and $\vec{\pi}_{x_2}$. In fact, the first signaling scheme is a mean-preserving spread of the second one. The seller, intuitively, wishes to design a menu with three options, one for each posterior.

This third posterior induces a trade-off in the linear program the seller solves to find the optimal mechanism. The second linear program has more IC constraints for
the two posteriors than the linear program given the fully revealing signaling scheme. This makes the revenue the seller gets from bidders with posteriors
$\vec{\pi}_{\partsig_1}$ and $\vec{\pi}_{\partsig_2}$ lower than before. The trade-off is that there is now a new posterior $\vec{\pi}_{\partsig_3}$, from which the seller can make additional revenue. Example~\ref{ex:nonmonotone} is constructed so that the harm from the additional posterior exceeds the benefit.

However, the lemma below shows that, when the data provider perfectly knows the item type, he minimizes the expected revenue of the seller by fully revealing said type. 

\begin{lemma}
If the data provider is adversarial and has full information about the type of the item (that is, if $X$ is perfectly correlated with the item type), the optimal strategy for the data provider is to reveal $X$, i.e., to use the fully-revealing signaling scheme~$S^*$.  
\end{lemma}

\begin{proof}
Let $SREV$ be the optimal revenue that the seller can achieve when the type of the item is revealed. On the one hand, the seller can always guarantee a revenue of $SREV$ by revealing the type of the item and then selling this type optimally, no matter what signaling scheme is used by the data provider.  So $SREV \leq REV(S)$ for any~$S$. On the other hand, when $X$ fully reveals the type, and $S^*$ fully reveals this information, the buyer would know the type, and by the definition of~$SREV$ the optimal revenue that can be achieved by the seller is $SREV$.  Therefore $SREV \geq REV(S^*)$.  Therefore $REV(S^*) \leq REV(S)$ for any scheme~$S$.
\end{proof}

Finally, we characterize the revenue the seller loses due to the presence of an adversarial data provider. Note that the revenue the seller can obtain when there is an adversarial data provider is less than what would be achieved under a strategic, revenue-maximizing data provider (who uses a fully-revealing signaling scheme). Thus, it follows from Corollary~\ref{cor: revenue_loss_strategic} that the presence of an adversarial data provider can greatly harm the revenue of the seller. 

\begin{corollary}\label{cor: revenue_loss_adversarial}
There exists a single seller, single bidder, single item setting where, in the presence of an adversarial data provider, no item-type partition mechanism can achieve revenue higher than $O\left(\frac{1}{\log n}\right)$ of the optimal revenue achievable by an item-type partition mechanism when no data provider is present. More precisely:
\begin{itemize}
\item In the absence of a data provider, the optimal revenue of the seller is $\Theta\left(\log n\right)$, and is attained by item-type bundling.
\item In the presence of a data provider, the optimal revenue is $1$, and is attained by item-type pricing.
\end{itemize}
\end{corollary}

\section{Concluding remarks}

Motivated by the increasing prominence of data markets, we study the impact of third-party data providers on the design of markets.  Specifically, we study the impact third-party data providers have on the ability of simple mechanisms to achieve a constant fraction of the optimal revenue.  This study is inspired by practical settings such as ad auctions.  In ad auctions a seller wants to sell an advertising slot to advertisers, but said advertisers may not fully understand the characteristics of the audience associated with the slot. However, the advertisers may be able to access additional information, such as cookie data, about said audience from third-party providers.  The question our results ask is whether the existence of third-party providers negatively impact the seller through a need for additional complexity in the auctions or revenue loss.  Our results show that yes, the presence of third-party data provider outside of the control of the seller can significantly hurt the seller's revenue and requires additional complexity in the auction design. As such, there are significant incentives for the seller to maintain a monopoly on the information available about the good for sale (e.g. ad slots).  This may explain why it is beneficial for companies such as Google provide their own cookie-matching services when running ad auctions, which limits the ability of third-party data providers to give new data to advertisers. 

In fact, our paper gives a  fine-grained characterization of what leads to losses in the seller's revenue when a third-party data provider is present. An obvious reason for lost revenue is that releasing additional information to the buyers participating in an auction allows them to bid more efficiently, which in turns drives down the optimal revenue the seller can hope to achieve. However, our results highlight that this is not the sole reason why the presence of a data provider reduces the seller's revenue.  Additionally, the seller's revenue loss is exacerbated by the use of simple, practical, and implementable mechanisms, which we show fail to capture the optimal achievable revenue in the presence of a third-party data provider. This is surprising, as in the mechanism design literature one of the main motivation behind the simple mechanisms we study is their ability to nearly capture the optimal revenue of the seller~\cite{BILW14,DPT16}. Given that it is desirable to run simple and practical mechanisms in real-life auctions, this raises the question of whether there exist alternative classes of ``simple'' mechanisms that achieve near-optimal revenue in the presence of a third-party data provider.  The search for such mechanisms represents an important direction for future research.

\section{Acknowledgments}
Cai thanks the Sloan Foundation for its support through a Sloan Foundation Research Fellowship. Part of Cai's work was done under the support of the NSERC Discovery grant RGPIN-2015-06127 and the FRQNT grant 2017-NC-198956. Echenique thanks the National Science Foundation for its support through grants SES-1558757 and CNS-1518941. Fu thanks the NSERC for its support through Discovery grant RGPAS-2017-507934 and Accelerator grant RGPAS-2017-507934. Ligett's work was supported in part by NSF grants CNS-1254169 and CNS-1518941, US-Israel Binational Science Foundation grant 2012348, Israeli Science Foundation (ISF) grant 1044/16, the United States Air Force and DARPA under contracts FA8750-16-C-0022 and FA8750-19-2-0222, and the HUJI Cyber Security Research Center in conjunction with the Israel National Cyber Directorate (INCD) in the Prime Ministers Office. Wierman thanks the National Science Foundation for its support through grants NSF AitF-1637598, CNS-1518941, as well as the Linde Institute of Economic and Management Science at Caltech. Ziani thanks the National Science Foundation for its support through grants CNS-1331343 and CNS-1518941, the US-Israel Binational Science Foundation through grant 2012348, and the Linde Graduate Fellowship at Caltech. We thank Noam Nisan for extremely useful comments and discussions.

\bibliographystyle{plainnat}
\bibliography{references}

\appendix

\section{Reducing multi-item auctions to single-item, multi-type auctions: an example}\label{app: reduction}

Consider a single bidder, single item setting with $n$ possible types, prior $\pi$ and valuation vector $\vec{V} = \left(V(1), \ldots, V(n) \right)$, distributed according to joint distribution $D$.~\citet{DPT16} show that this setting in fact reduces to a multi-item auction with $n$ items, in which the bidder's valuation for the items are distributed as follows: i) draw $\vec{V}$ according to $D$, then ii) let the bidder's valuation vector be $\vec{V'} = \left(\pi(1) V(1), \ldots,\pi(n) V(n) \right)$, the coordinate-by-coordinate product of $\vec{V}$ and $\pi$. 

The optimal auction in such a single-bidder, multi-item setting can be written without loss of generality as a menu of options $\{(P_o,\vec{A}_o)\}_{o}$, such that a bidder either opts out, or selects a single option $o$ in which case he i) must pay price $P_o$ then ii) receives any given item $i \in [n]$ with allocation probability $A_o(i)$. An optimal single-bidder, single-item, multi-type auction is then given by the exact same menu $\{(P_o,\vec{A}_o)\}_{o}$, where i) a bidder who picks option $o$ picks price $P_o$, but now ii) the bidder receives the (here, a single) item with probability $A_o(i^*)$ where $i^*$ is the realized item type.
For example, consider the following scenario, studied previously in~\citet{hart2015max}: 
\begin{example}\emph{
There are two item types, denoted $1$ and $2$. The bidder's prior on the types is uniform, given by $\pi(1)=\pi(2)=1/2$. The bidder's valuations $V(1)$ and $V(2)$ are i.i.d, and take values $2$, $4$ and $8$ with probabilities $1/6$, $1/2$ and $1/3$.
The equivalent multi-item auction is one with two items whose valuations $V'(1),V'(2)$ are i.i.d, and take values distributed as follows: $V'(1) = \pi(1) V(1) = V(1)/2$ and $V'(2) = \pi(2) V(2) = V(2)/2$; i.e., $1$, $2$ and $4$ with probabilities $1/6$, $1/2$ and $1/3$. As shown by~\citet{hart2015max}, the optimal mechanism in this case has two menu options:
\begin{enumerate}
\item Option $1$ has price $1$, and gives the first item with probability $1/2$ and the second with probability $0$.
\item Option $2$ has price $4$, and gives both items with probability $1$.
\end{enumerate}
This translates into a menu of options in the multi-type, single-item settings in which a bidder that selects option $1$ gets the item only if it is of type $1$, with probability $1/2$, and a bidder that selects option $2$ always get the item, independently of the item type. 
}
\end{example}

\section{Characterization of the optimal mechanism and the proof of Theorem~\ref{LEM:REVPR}}\label{appx:charac}

The characterization we present shows that the revenue achievable via any mechanism can be obtained with a conditional price menu.

A few comments about the mechanism are in order. Note that the allocation probability and price may both depend on the realized type of the item. So, one can think of the mechanism as requiring a single round of bidding, followed by a single round of information revelation (in fact, full information revelation), to determine which price $P(i)$ the bidder should pay. Additionally, note that the bidder pays regardless of whether he receives the item.  Finally, note that conditional price mechanisms are strictly more general than item-type partition mechanisms. Item-type partitioning is, in fact, an instantiation of menus with conditional prices in which each $\vec{Z} \in \{0,1\}^n$ (no fractional or probabilistic allocations are allowed), each item type is offered in exactly one option, and the conditional prices within an option are all identical. Each option then corresponds to a single subset of the partition.

Despite allowing prices and allocations to depend on the realization of the item type, the conditional price menus guarantee interim individual rationality, defined as follows.

\begin{definition}
\emph{A mechanism is \textbf{interim individually rational} (interim IR) if and only if the bidder's expected utility from participating in the mechanism, conditional on a valuation $\vec{V}$ and posterior beliefs $\vec{\pi}$ over item types, is non-negative.}
\end{definition}

Interim IR can be seen as the bidder committing to an option from the menu offered by the mechanism. One justification for this notion is that a bidder might, in theory, be engaged in many auctions simultaneously. Therefore, the bidder might care only about his average payoff across multiple purchases. While for some type realizations such a bidder may lose, with high probability his overall utility is non-negative. Interim IR can always be guaranteed by adding a dummy option with price $0$ and allocation probability $\vec{Z} = \vec{0}$, such that an agent that gets negative utility from any other option goes for the dummy option. 

\begin{proof}[Proof of Theorem~\ref{LEM:REVPR}] We treat the pair $(\vec{V},\vec{\pi}_s)$, where $\vec{V}$ is the bidder's valuation vector and $\vec{\pi}_s$ is his posterior given he sees signal $s$, as the bidder's type. 
We follow the same steps as the proof of Theorem 1 and Appendix A of~\citet{DPT16}. Consider a mechanism $\cM$ with voluntary participation. $\cM$  may use multiple rounds of communication and information revelation to the bidder. For each valuation vector $\vec{V}$ and posterior $\vec{\pi}_s$, let $A\left( \vec{V},\vec{\pi}_s \right)$ be the (possibly randomized) equilibrium strategy of the bidder when his type is $\left( \vec{V},\vec{\pi}_s \right)$.

Let $Z(i, A)$ be an indicator random variable that indicates whether the bidder gets the item when he chooses strategy $A$ and the realized item type is $i$. Similarly, let $C(i,A)$ denote the price the bidder is asked to pay. The bidder's interim expected utility is then given as follows:
$$
\E_{i \sim\vec{\pi}_s} \left[ \E\left[Z(i, A) \cdot V(i) - C(i, A)\right]  \right]
$$
where the first (outer) expectation is with respect to the randomness of the item type, while the second (inner) expectation is with respect to the randomness in the choices of the mechanism, the information revealed and the actions $A$ of the bidder.

For all possible types $\left( \vec{V},\vec{\pi}_s \right)$, and for all possible misreports $\left(\vec{V'},\vec{\pi}_{s'}\right)$ of the bidder, for $A$ to be an equilibrium strategy it must be the case that 
\begin{align*}
& \E_{i \sim\vec{\pi}_s}\left[ \E \left[Z \left(i, A\left( \vec{V},\vec{\pi}_s \right)\right) V(i) - C\left(i, A\left( \vec{V},\vec{\pi}_s \right)\right)\right]\right] 
\\ &\geq 
\E_{i \sim\vec{\pi}_s}\left[ \E \left[Z \left(i, A\left( \vec{V'},\vec{\pi}_{s'} \right) \right) V(i) - C\left(i, A\left( \vec{V'},\vec{\pi}_{s'} \right)\right)\right]\right]. 
\end{align*}

Now define the variables
\begin{align*}
&z_{i}\left( \vec{V},\vec{\pi}_s \right) = \E \left[Z \left(i, A\left( \vec{V},\vec{\pi}_s \right)\right) \right]
\\&c_{i}\left( \vec{V},\vec{\pi}_s \right) = \E \left[C\left(i, A\left( \vec{V},\vec{\pi}_s \right)\right)\right].
\end{align*}
The above equation can be rewritten as 
\begin{align}
\tag{IC}
\sum\limits_i \pi_s(i) \left( z_{i}\left( \vec{V},\vec{\pi}_s \right) V(i) - c_i\left( \vec{V},\vec{\pi}_s \right) \right)
\geq 
\sum\limits_i \pi_s(i) \left( z_{i}\left( \vec{V'},\vec{\pi}_{s'} \right) V(i) - c_i\left( \vec{V'},\vec{\pi}_{s'} \right) \right).
\end{align} 

Moreover, since the equilibrium $A$ respects voluntary participation, the bidder's equilibrium payoff must be non-negative. As a consequence, we have
\begin{align}
\tag{IR}
\sum\limits_i \pi_s(i) \left( z_{i}\left( \vec{V},\vec{\pi}_s \right) V(i) - c_i\left( \vec{V},\vec{\pi}_s \right) \right)
\geq 0.
\end{align}

Finally, we note that the revenue of the seller is given by
\begin{align*}\label{eq: rev}
R = \sum_{\vec{\pi}_s,\vec{V}} \Pb\left[ \vec{V},\vec{\pi}_s\right] \sum_i \pi_s(i) \cdot c_i\left( \vec{V},\vec{\pi}_s \right), 
\end{align*}
where $\Pb\left[\vec{V},\vec{\pi}_s \right]$ is the probability the realized type of the bidder is $\left(\vec{V},\vec{\pi}_s \right)$. 

A mechanism that satisfies constraints (IC) and (IR) and yields revenue $R$ can clearly be implemented as an interim IR menu with conditional prices, with options given by 
$
\left(\vec{z}\left( \vec{V}, \vec{\pi}_s \right),\vec{c} \left( \vec{V}, \vec{\pi}_s \right) \right)
$
for each possible type $\left( \vec{V}, \vec{\pi}_s \right)$. Hence, there exists an incentive compatible, individually rational, conditional price menu that provides the same revenue as mechanism $\cM$. 
\end{proof}

We now provide a characterization of the optimal mechanisms in the presence of several bidders. We consider a setting in which different bidders may have different priors on the item type, and may receive different, private signals. In this case, a bidder only has knowledge of his own type -- defined as the combination of his private valuation vector $V_j$ and private signal $s_j$ --, and his strategy is a function $A_j(V_j,s_j)$ of his private type. We will show that for every mechanism $\cM$ and every ex-post Nash Equilibrium strategy in said mechanism $\cM$, there exists an interim individually rational and ex-post incentive compatible direct revelation mechanism that achieves the same revenue. Further, the best such direct revelation mechanism is the solution to a linear program (of infinite size in the number of possible agent types and signals, and exponential size in the number of agents). While we focus on the ex-post case, we note that the result extends to the Bayesian case, and is obtained via a similar proof; i.e., for every mechanism and Bayes-Nash equilibrium of said mechanism, there is a corresponding direct revelation mechanism that is Bayes incentive compatible and achieves the same revenue. 

Formally, we first define what an ex-post Nash Equilibrium of a mechanism is:

\begin{definition}
\emph{
Consider a mechanism $\cM$ with $b$ agents. Let $u_j\left(\vec{T}_j,\vec{A}\right)$ be the utility obtained by agent $j$, when his type is $\vec{T}_j$ and when the agents' (possibly randomized) strategies are given by $\vec{A} = \left(A_1,\ldots,A_b\right)$. A strategy profile $\vec{A^*} = \left(A^*_1,\ldots,A_b\right)$ is an ex-post Nash Equilibrium for mechanism $\cM$ if and only if for every agent $j$, for all possible types $\vec{T}_{-j}$ of the remaining agents and any alternate strategy $A_j'$ for bidder $j$, 
\[
u_j\left(\vec{T}_j,\left(A^*(\vec{T}_1),\ldots,A^*(\vec{T}_j) \right)\right) \geq u_j\left(\vec{T}_j,\left(A_j',A_{-j}^*(\vec{T}_{-j}) \right)\right).
\]
A direct-revelation mechanism $\cM$ with $b$ agents is ex-post incentive compatible if it is an ex-post Nash Equilibrium for each agent to truthfully report his type to the mechanism. 
}
\end{definition}

Informally, an ex-post Nash Equilibrium is a strategy in which an agent $j$ does not want to deviate from his current strategy, \emph{no matter what the types of the other agents are, as long as they play according to the equilibrium strategy $A^*$}. I.e., even \emph{after} (hence ``ex-post'') observing the other agents' types, $j$ cannot improve his expected utility (where the expectation is taken over the randomness of the strategies) by changing his strategy. Similarly, a direct revelation is \emph{ex-post} incentive compatible if and only if truthfully reporting his type maximizes his utility from participating in the mechanism, no matter what the true types of the other agents are, assuming they all truthfully report their type. We can now write the main result:
\begin{theorem}
For any ex-post Nash Equilibrium of any mechanism $\cM$ with $b$ buyers in the presence of a data provider, such that every buyer, conditioned on the realization of his own valuation vector and posterior belief over item types given the signal from the data provider, obtains non-negative payoff in expectation, there is a direct revelation mechanism that is ex-post incentive compatible, interim individually rational, and provides the same revenue. Further, a revenue-maximizing incentive compatible and interim individually rational direct revelation mechanism can be written as the solution of a linear program.
\end{theorem}

\begin{proof}
Once again, we follow the same steps as the proof of Theorem 1 and Appendix A of~\citet{DPT16}. For every bidder $j \in [b]$, we treat the pair $(\vec{V}_j,\vec{\pi}_{s_j})$, where $\vec{V}$ is bidder $j$'s valuation vector and $\vec{\pi}_{s_j}$ is his posterior when he sees signal $s_j$, as the bidder's type. For simplicity of notations, we often write $\vec{T}_j = \left( \vec{V}_j,\vec{\pi}_{s_j} \right)$ in the rest of the proof. Consider a mechanism $\cM$ with voluntary participation. $\cM$  may use multiple rounds of communication and information revelation to the bidder. For each type $\vec{T}_j$, let $A_j\left( \vec{T}_j \right)$ be the (possibly randomized) equilibrium strategy of the bidder when his type is $\vec{T}_j$.

Let $Z_j(i, A)$ be an indicator random variable that indicates whether bidder $j$ gets the item when buyers choose strategies $A = \left(A_1, \ldots, A_b\right)$ and the realized item type is $i$. Similarly, let $C_j(i,A)$ denote the price the bidder is asked to pay. Bidder $j$'s interim expected utility is then given as follows:
$$
\E_{i \sim\vec{\pi}_{s_j}} \left[ \E\left[Z_j(i, A) \cdot V_j(i) - C_j(i, A)\right]  \right]
$$
where the first (outer) expectation is with respect to the randomness of the item type, while the second (inner) expectation is with respect to the randomness in the choices of the mechanism, the information revealed and the actions $A$ of the bidders, and holds the bidders valuations and signals fixed. Note that the agents' valuations are fixed. 

Let $A(\vec{T}) \triangleq \left(A_1(\vec{T}_1),\ldots,A_b(\vec{T}_b)\right)$. For $A$ to be an ex-post equilibrium strategy, it must in particular be the case for all possible bidder types $(\vec{T}_1,\ldots,\vec{T}_b)$ and for all possible misreports $\vec{T}'_j = \left( \vec{V}'_j,\vec{\pi}_{s'_j} \right)$ for bidder $j$ that, 
\begin{align*}
& \E_{i \sim\vec{\pi}_{s_j}}\left[ \E \left[Z_j \left(i, A(\vec{T})\right) V_j(i) - C_j\left(i, A(\vec{T})\right)\right]\right] 
\\ &\geq 
\E_{i \sim\vec{\pi}_{s_j}}\left[ \E \left[Z_j \left(i, A\left( \vec{T}'_j,\vec{T}_{-j}\right) \right) V_j(i) - C_j\left(i, A\left( \vec{T}'_j,\vec{T}_{-j}\right) \right)\right]\right], 
\end{align*}
noting that $A'_j = A_j(\vec{T}'_j)$ is a possible deviation for agent $j$, where the outer expectation is over $i$, the inner expectation is over the randomness in the choices of the mechanism and the actions $A$ of the bidder, and the agents' valuations and received signals are fixed.
\begin{align*}
&z_{ij}\left( \vec{T}_j,\vec{T}_{-j} \right) = \E \left[Z_j \left(i, A\left( \vec{T}_j,\vec{T}_{-j} \right)\right) \right]
\\&c_{ij}\left( \vec{T}_j,\vec{T}_{-j} \right) = \E \left[C_j\left(i, A\left( \vec{T}_j,\vec{T}_{-j} \right)\right)\right].
\end{align*}
The above equation can be rewritten as 
\begin{align}
\tag{IC}
\sum\limits_i \pi_{s_j}(i) \left( z_{ij}\left( \vec{T}_j,\vec{T}_{-j} \right) V_j(i) - c_{ij}\left( \vec{T}_j,\vec{T}_{-j} \right) \right)
\geq 
\sum\limits_i \pi_{s_j}(i) \left( z_{ij}\left( \vec{T}'_j,\vec{T}_{-j} \right) V_j(i) - c_{ij}\left( \vec{T}'_j,\vec{T}_{-j} \right) \right).
\end{align} 

Moreover, since the equilibrium $A$ respects voluntary participation, every bidder's equilibrium payoff must be non-negative. As a consequence, we have
\begin{align}
\tag{IR}
\sum\limits_i \pi_{s_j}(i) \left( z_{ij}\left( \vec{T}_j,\vec{T}_{-j} \right) V_j(i) - c_{ij}\left( \vec{T}'_j,\vec{T}_{-j} \right) \right)
\geq 0.
\end{align}
Finally, we note that the revenue of the seller is given by
\begin{align*}\label{eq: rev}
\tag{Rev}
R 
&=  \sum_{\{\vec{T}_j\}_{j \in [b]}} 
\Pb\left[ \{\vec{T}_j\}_{j \in [b]}\right] \sum_{j=1}^b \sum_i \pi_{s_j}(i) \cdot c_{ij}\left( \vec{V}_j,\vec{\pi}_{s_j} \right), 
\end{align*}
where $\Pb\left[ \{\vec{T}_j\}_{j \in [b]}\right]$ is the probability that for all $j$, bidder's $j$ realized type is $\vec{T}_j$.

A mechanism that satisfies constraints (IC) and (IR) and yields revenue $R$ can clearly be implemented as a direct revelation mechanism that is ex-post IC and interim IR, that works as follows: first, every agent $j$ reports his valuation vector $\vec{V}_j$ and posterior $\vec{\pi}_{s_j}$ to the seller, then the seller reveals the item type $i$, and allocates to bidder $j$ with probability $z_{ij}\left( \vec{T}_1 \ldots,\vec{T}_b \right)$ at price $c_{ij}\left( \vec{T}_1 \ldots,\vec{T}_b \right)$, where $\vec{T}_j = \left(\vec{V}_j,\vec{\pi}_{s_j}\right)$. Further, a revenue-maximizing such mechanism can be found by solving the linear program that has variables $z_{ij}\left( \vec{T}_1 \ldots,\vec{T}_b \right),~c_{ij}\left( \vec{T}_1 \ldots,\vec{T}_b \right)$ for all item type $i$, bidder $j$, and bidder type $T_j$, objective (Rev), and constraints (IC), (IR).
\end{proof}

\section{Proof of Theorem~\ref{THM:SIMPLE}}\label{app:simple}

\begin{proof}[Proof of Lemma~\ref{lem:item_simple}]
In item-type pricing, the seller announces the item type (hence, completely superseding the effect of the data provider's signal)
and then offers a price that is a function of the realized item type. The expected revenue of such a mechanism is simply given by 
\[
\frac{1}{n} \sum_{k=1}^m \frac{m}{k} = \Theta \left( \frac{\log n}{\sqrt{n}} \right), 
\]
as the expected revenue from selling an item of type $i$ in the $k$th group is $P \cdot \Pr \left[V(i) \geq P \right] = P \cdot \frac{1}{kP}= \frac{1}{k}$, as $k \cdot V(i)$ follows an ER distribution. 
\end{proof}

\begin{proof}[Proof of Lemma~\ref{lem:bundling_simple_noDP}]
This proof follows the same structure as the proof of Proposition 25 of \citet{HN12}. For all $i$ and all $M \geq 1$, we let $V^M(i) = \min \left(V(i),M \right)$. By~\citet{HN12}, $V^M(i)$ has mean $\log M + 1$ and variance upper-bounded by $2M$. In particular, it follows that the expectation and variance of the value of the bundle (renormalized by $n$), were the bidders' valuations truncated at $M$, satisfy
\begin{align*}
&\E \left[ \sum_{k = 1}^m \sum_{i \in I_k} \frac{V^M(i)}{k}\right] 
= \left( \log M + 1 \right) \sum_{k=1}^m \frac{m}{k} 
\\&\in \left[ \frac{1}{2} \left(\log M + 1 \right) \sqrt{n} \log n; \left(\log M + 1 \right) \sqrt{n} \left( 1 + \frac{1}{2} \log n \right) \right],
\end{align*}
and
\begin{align*}
\Var\left[ \sum_{k = 1}^m \sum_{i \in I_k} \frac{V^M(i)}{k}\right]  
\leq  2 M m \sum_k \frac{1}{k^2} \leq \frac{\pi^2}{3} M \sqrt{n}.
\end{align*}

We first give a lower bound on the revenue of the item-type bundling mechanism. 
\begin{align*}
&\Pr \left[ \frac{1}{n} \sum_{k = 1}^m \sum_{i \in I_k} \frac{V(i)}{k} \geq P \right]
\\&\geq \Pr \left[\sum_{k = 1}^m \sum_{i \in I_k} \frac{V^M(i)}{k} \geq n P \right]
\\&= \Pr \left[ \sum_{k = 1}^m \sum_{i \in I_k} \frac{\E \left[V^M(i)\right] - V^M(i)}{k} 
\leq m \left( \log M + 1 \right) \cdot \sum_{k=1}^m \frac{1}{k} - n P \right]
\\&= 1 
- \Pr \left[ \sum_{k = 1}^m \sum_{i \in I_k} \frac{\E \left[V^M(i)\right] - V^M(i)}{k}
>\left( \log M + 1 \right) \sum_{k=1}^m \frac{m}{k} - n P \right]
\\&\geq 1 - \frac{\pi^2 M \sqrt{n}}{3\left(\left( \log M + 1 \right)\sum_{k=1}^m \frac{m}{k} - nP\right)^2}
\\& \geq 1 - \frac{\pi^2 M \sqrt{n}}{3\left( \frac{1}{2} \left(\log M + 1 \right) \sqrt{n} \log(n) - nP\right)^2} 
\end{align*}
where the second-to-last step follows from Chebyshev's inequality, when $m \left( \log M + 1 \right) \cdot \sum_{k=1}^m \frac{1}{k} - n P  \geq 0$. Let $M = \sqrt{n} \log^2 n$ and $P = \frac{\log^2 n}{4\sqrt{n}}$, we obtain that 
\begin{align*}
&\Pr \left[ \frac{1}{n} \sum_{k = 1}^m \sum_{i \in I_k} \frac{V(i)}{k} \geq P \right]
\\&\geq 1 - \frac{\pi^2 n \log^2 n}{3\left( \frac{1}{2} \left(\frac{1}{2}\log n  + 2\log \log n+ 1 \right) \sqrt{n} \log n - \frac{\sqrt{n} \log^2 n}{4} \right)^2}
\\& \geq 1 - \frac{\pi^2 n \log^2 n}{3 \left( \sqrt{n} \log n \cdot \log \log n\right)^2}
\\&\geq 1 - \frac{\pi^2}{3 \left( \log \log n \right)^2}
\end{align*} 
Therefore, a buyer buys a bundle with price $P = \frac{\log^2 n}{4\sqrt{n}}$ with constant probability (for $n$ large enough), guaranteeing a revenue of $\Omega \left( \frac{\log^2 n}{\sqrt{n}}\right)$. 

For the upper bound, we first remark that $P \leq 2  \frac{\log n }{\sqrt{n}} \left(1 + \frac{\log n}{2} \right)$ implies that the revenue is at most $O \left( \frac{\log^2 n}{\sqrt{n}} \right)$. We therefore assume w.l.o.g that $P > 2  \frac{\log n }{\sqrt{n}} \left(1 + \frac{\log n}{2} \right)$. The revenue from grand bundling at price $P$ satisfies, by union bound:
\begin{align*}
P \cdot \Pr \left[ \frac{1}{n} \sum_{k = 1}^m \sum_{i \in I_k} \frac{V(i)}{k} \geq P \right] 
&\leq P \cdot \Pr \left[ \sum_{k = 1}^m \sum_{i \in I_k} \frac{V^{n P}(i)}{k} \geq n P \right] 
\\& + P \cdot \Pr \left[\exists k, i \in I_k:~\frac{V(i)}{k} \geq n P  \right]
\end{align*}

By union bound, we have on the one hand that 
\begin{align*}
P \cdot \Pr \left[\exists k, i \in I_k:~\frac{V(i)}{k} \geq n P  \right] 
&\leq P \sum_k \sum_{i \in I_k} \Pr \left[ V(i) \geq k nP\right] 
\\&= m P \sum_k \frac{1}{k n P} 
\\&= O \left( \frac{\log n}{\sqrt{n}}\right).
\end{align*}
On the other hand, remembering that 
\begin{align*}
&\E \left[ \sum_{k = 1}^m \sum_{i \in I_k} \frac{V^{n P}(i)}{k}  \right] \leq (\log(n P) + 1) \sqrt{n} \left(1 + \frac{\log n}{2} \right)
\\&\Var \left[ \sum_{k = 1}^m \sum_{i \in I_k} \frac{V^{n P}(i)}{k}  \right] \leq \frac{\pi^2}{3} n \sqrt{n} P,
\end{align*}
we have by Chebyshev that 
\begin{align*}
&P \cdot \Pr \left[ \sum_{k = 1}^m \sum_{i \in I_k} \frac{V^{n P}(i)}{k} \geq n P \right] 
\\&\leq \frac{\pi^2 n \sqrt{n} P^2}{3 \left(n P - (\log n + \log P + 1) \sqrt{n} \left(1 + \frac{\log n}{2} \right)\right)^2}.
\end{align*}
Using the fact that w.l.o.g, $P \geq 2  \frac{\log n }{\sqrt{n}} \left(1 + \frac{\log n}{2} \right)$ or equivalently $\sqrt{n} \log n \left(1 + \frac{\log n}{2} \right) \leq nP/2$, and that $ \left(\log P + 1\right) \sqrt{n} \left(1 + \frac{\log n}{2} \right) = o \left(nP \right)$ we have that 
\begin{align*}
&n P - (\log n + \log P + 1) \sqrt{n} \left(1 + \frac{\log n}{2} \right) 
\\&= n P - \left(\log P + 1\right) \sqrt{n} \left(1 + \frac{\log n}{2} \right) - \sqrt{n} \log n \left(1 + \frac{\log n}{2} \right)
\\& \geq nP - o (nP) - \frac{nP}{2}
\\& = \Omega \left(nP\right),
\end{align*}
which leads to 
\begin{align*}
&P \cdot \Pr \left[ \sum_{k = 1}^m \sum_{i \in I_k} \frac{V^{n P}(i)}{k} \geq n P \right] = O \left( \frac{n \sqrt{n} P^2}{n^2 P^2}\right) = O \left( \frac{1}{\sqrt{n}}\right).
\end{align*} 
Hence, 
\[
P \cdot \Pr \left[\exists k \in [m], i \in I_k:~\frac{V(i)}{k} \geq n P  \right] = O \left( \frac{\log n}{\sqrt{n}}\right),
\]
which concludes the proof.
\end{proof}

\begin{proof}[Proof of Lemma~\ref{lem:bundling_simple}]
Let $P^*$ be the optimal bundling price, and suppose the data provider announces signal $s_k$. There are two cases:
\begin{enumerate}
\item For $k$ such that $P^* \geq \frac{6}{k} \log m$, by Lemma~\ref{lem: ER_property}, the expected revenue is
\begin{align*}
P^* \cdot \Pb \left[\frac{1}{m} \sum_{i \in |I_k|}  \frac{V(i)}{k} \geq P^* \right] 
&= P^* \cdot \Pb \left[\frac{1}{m} \sum_{i \in |I_k|}  V(i) \geq k P^* \right] 
\\&\leq P^* \cdot \frac{9}{k P^*} 
\\&= \frac{9}{k},
\end{align*}
as 
$|I_k| = m$.
\item Otherwise, we have $k$ such that $P^* \leq \frac{6}{k} \log m$.
\end{enumerate}
Letting $k^* = \min \{k:~P^{*} > \frac{6}{k} \log m\}$, we see that the expected revenue of charging price $P^*$ for the grand bundle is upper-bounded by
\begin{align*}
\frac{1}{m} \left( \sum_{k \geq k^*} \frac{9}{k} + \sum_{k < k^*} P^* \right) 
& \leq \frac{ 9 \cdot (1 + \log m) + \sum_{k < k^*} \frac{6}{k^* - 1} \log m}{\sqrt{n}}
\\&= \frac{ 9 \cdot (1 + \log m) + 6 \log m}{\sqrt{n}}
\\&= O \left(\frac{\log n}{\sqrt{n}} \right).\qedhere
\end{align*}
\end{proof}

\begin{proof}[Proof of Lemma~\ref{lem:rev_opt}]

Consider the following item-type partition mechanism: the seller first partitions the item types into $m$ groups in the same way as specified in Construction~\ref{ex:simple}. When the realized item type is in group $I_k$, she offers to sell the item to the bidder at price $P_k=\frac{\log m}{2k}$.

If the bidder receives signal $s_k$, then the price offered by the seller must be $ \frac{\log m}{2k}$, and the bidder knows the item type is from group $I_k$. By Lemma~\ref{lem: ER_property}, as $|I_k| = m$, we have:
\[
\Pb \left[ \frac{1}{m} \sum_{i \in I_k} \frac{V(i)}{k} \geq \frac{\log m}{2k} \right] \geq \frac{1}{2},
\]
and hence with probability at least $1/2$, conditional on $S=s_k$, he accepts the price, yielding expected revenue to the seller of at least $\frac{\log m}{4k}$. The total expected revenue for the seller is then given by
\[
\frac{1}{m} \sum_{k=1}^m \frac{\log m}{4k} = \frac{\log n}{8 \sqrt{n}} \sum_{k=1}^m \frac{1}{k} = \Omega\left( \frac{\log^2 n}{\sqrt{n}} \right).
\]

No truthful mechanism can achieve revenue higher than $\log m$ times the revenue of item-type pricing conditioned on receiving signal $s_k$: Theorem 2 of~\citet{LiY13} shows that in traditional multi-item auctions, selling separately achieves at least a $\Omega\left(\frac{1}{\log m}\right)$ fraction of the optimal revenue for selling $m$ independent items; this result carries over to single-item, multi-type auctions by the reduction of~\citet{DPT16}. Thus, the optimal revenue is at most $O \left( \frac{\log^2 n}{\sqrt{n}} \right)$, and hence the item-type partition mechanism we just described yields a constant approximation to the optimal revenue.
\end{proof}

\section{Proof of Theorem~\ref{THM:ITEM-TYPE-PARTITION}} \label{app:item-type-partition}

\begin{proof}[Proof of Lemma~\ref{lemma:partition}]

Suppose the item-type partition mechanism splits the item types into non-empty groups $\cG_1$ to $\cG_{g}$, where $g \leq 2^m$ is the number of such groups. Let us assume that the item type $i$ lies in $\cG_r$, then the seller offers to sell the item at price at $P_r$. Suppose the signal is $s_{k,j}$ for some $j\in[n_k]$ with $i\in  \cI_{k,j}$. In the bidder's posterior, the item type is uniform over $\cG_{r} \cap \cI_{k,j}$. Note that $\left |\cG_{r} \cap \cI_{k,j} \right| \leq 2^k$. By Lemma~\ref{lem: ER_property}, we have
\[
P_{r} \cdot \Pb \left[\frac{1}{\left|\cG_{r} \cap \cI_{k,j}\right|} \sum_{t\in \cG_{r} \cap \cI_{k,j}} V(t) \geq P_{r} \right] \leq\begin{cases}
9 & \text{ if }  P_{r}  \geq 6k \log 2,\\
P_{r}  & \text{ if } P_{r} < 6k \log 2,\\
\end{cases}
\] 
following from $6 \log \left( 2^{k} \right) \geq 6 \log \left|\cG_{r} \cap \cI_{k,j}\right| $.

Let $k^*(r) = \max \{k:~P_{r} \geq 6k \log 2\}$. Further, let us denote by $ \Pb[k]$ the probability that the data provider selects a partition of size $2^k$. When item type $i\in \cG_r$, the revenue in expectation over the randomness of the signal is upper-bounded by
\begin{align*}
&9 \sum_{k \leq k^*(r)}  \Pb [k] + P_{r} \cdot \left(\sum_{k = k^*(r)+1}^{m} \Pb [k]\right)  
\\&= 9\sum_{k \leq k^*(r)} \frac{1}{k(k+1)} +  P_{r}\cdot \left(\sum_{k = k^*(r)+1}^{m-1} \frac{1}{k(k+1)} + \frac{1}{m} \right) \\
& \leq 9 \left(1 - \frac{1}{k^*(r)+1} \right) 
\\&+ 6\log 2 \cdot (k^*(r)+1)  \left(\frac{1}{k^*(r)+1}-\frac{1}{m} + \frac{1}{m} \right)
\\&\leq 9 + 6 \log 2, 
\end{align*}
where the first step follows from the fact that the probability of the data provider selecting a $k \leq m-1$ is $\frac{1}{k(k+1)}$, and the probability of him drawing $k=m$ is $\frac{1}{m}$. Since the upper bound holds for all possible prices, the expected revenue of any item-type partition mechanism is also upper-bounded by $9 + 6 \log 2$.
\end{proof}

\begin{proof}[Proof of Lemma~\ref{clm: inter_ex2}]
The bidder's expected utility for option $L_{k,j}$ when receiving signal $s_{k,j}$ is given by
\[
U_{k,j} = \frac{1}{2^{k}} \sum_{i \in I_{k,j}} V(i) - \frac{1}{8} \log 2^{k},
\]
his expected utility for selecting  option $L_{\kappa,\iota}$ for $\kappa > k$ is only less (his expected value for the item type is not more, but the price is higher),
and his utility for selecting option $L_{\kappa,\iota}$ for $\kappa < k$ is 
\[
U_{\kappa,\iota} = \frac{1}{2^{k}} \sum_{i \in I_{\kappa,\iota} \cap I_{k,j}} V(i) - \frac{1}{8} \log 2^{\kappa},\]
and
his expected utility for selecting any option $L_{\kappa,\iota}$ such that $I_{\kappa,\iota} \cap I_{k,j}=\emptyset$ is negative, since he will pay but never be allocated the item.

Therefore, the bidder prefers $L_{\kappa,\iota}$ to $L_{k,j}$ with $\kappa \leq k$ and $I_{\kappa,\iota} \subset I_{k,j}$ only if
\[
 \frac{1}{2^k} \sum_{i \in I_{k,j}\setminus I_{\kappa,\iota}} V(i) \leq \frac{1}{8} \log 2^{k-\kappa}.
\]

We want to upper bound the probability of the above event for all $\kappa < k$ and $I_{\kappa, \iota} \subset I_{k,j}$. Let us denote $V^M(i) = \min(V(i),M)$ for any $M$. We have immediately that $\E \left[V^M(i) \right] = \log M + 1$ and that its variance is upper-bounded by $2M$. Taking $M = 2^{k-1}$ and $W(i) = \log M+ 1 - V^M(i)$ yields $|W(i)| \leq M$, $\E \left[ W(i) \right] = 0$ and $\E \left[ W(i)^2 \right] \leq 2 \cdot 2^{k-1} = 2^{k}$. Recall Bernstein's inequality:

\begin{lemma}\label{lem:bernstein}
(Bernstein's Inequality): Suppose $X_1,...,X_n$ are independent random variables with zero mean, and $|X_i|\leq B$ almost surely for all $i$. Then for any $t>0$,
\[\Pr\left[\sum_{i=1}X_i>t\right]\leq exp\left(-\frac{\frac{1}{2}t^2}{\sum_{i=1}^nE[X_i^2]+\frac{1}{3}Bt}\right)\]
\end{lemma}

We can then apply Bernstein's inequality to show that
\begin{align*}\label{eq: Bernstein}
&\Pb \left[\frac{1}{2^k} \sum_{i \in I_{k,j}\setminus I_{\kappa,\iota}} V(i) 
      < \frac{1}{2^k} \cdot \left(\sum_{i \in I_{k,j}\setminus I_{\kappa,\iota}} \left( \log M + 1 \right) - t\right)  \right]\\
& =\Pb \left[ \sum_{i \in I_{k,j}\setminus I_{\kappa,\iota}} V(i) 
      < \sum_{i \in I_{k,j}\setminus I_{\kappa,\iota}} \left( \log M + 1 \right) - t  \right]\\
& \leq \Pb \left[ \sum_{i \in I_{k,j}\setminus I_{\kappa,\iota}} V^M(i) 
         < \sum_{i \in I_{k,j}\setminus I_{\kappa,\iota}} \left( \log M + 1 \right) - t  \right] \\
& = \Pb \left[ \sum_{i \in I_{k,j}\setminus I_{\kappa,\iota}} W(i) > t  \right] \nonumber\\
&\leq \exp \left(-\frac{1}{2} \cdot \frac{t^2}{2^{k} \cdot |I_{k,j} \setminus I_{\kappa,\iota} | + M \cdot t/3}\right) \\
& = \exp \left(-\frac{1}{2} \cdot \frac{t^2}{2^{k} \left( 2^k- 2^{\kappa} \right) + M \cdot t/3}\right), 
\end{align*}
where the last inequality is due to {Bernstein's inequality}. Taking
\begin{align*}
t= \left(\frac{3}{4}\right) \left( 2^k- 2^{\kappa} \right) \left( \log M + 1 \right),
\end{align*}
we have
\begin{align*}
\frac{1}{2^k}\left(\sum_{i \in I_{k,j}\setminus I_{\kappa,\iota}} \left( \log M + 1 \right) - t\right)  &=
\frac{1}{2^k}\cdot \frac{1}{4} \left(2^k-2^\kappa\right)\left(\log M+1\right) \\
&\geq \frac{1}{2^k} \cdot \frac{1}{4}\left(2^{k -1}\right)\left(\log M+1\right)\\
&=\frac{1}{8}\left(\log M+1\right),
\end{align*}
and we thus obtain a bound on the probability of the event that a particular menu option $L_{\kappa,\iota}$  for $\kappa < k$ is better for the bidder than option $L_{k,j}$, given signal $s_{k,j}$:
\begin{align*}
&\Pb\left[\frac{1}{2^k} \sum_{i \in I_{k,j}\setminus I_{\kappa,\iota}} V(i) \leq \frac{1}{8} \log 2^{k-\kappa} \right]
\\& <
\Pb \left[ \frac{1}{2^k}\sum_{i \in I_{k,j}\setminus I_{\kappa,\iota}} V(i) < \frac{1}{8} \left( \log 2^{k-1} + 1 \right) \right]\\
& \leq \Pb \left[\frac{1}{2^k} \sum_{i \in I_{k,j}\setminus I_{\kappa,\iota}} V(i) 
      < \frac{1}{2^k} \left(\sum_{i \in I_{k,j}\setminus I_{\kappa,\iota}} \left( \log 2^{k-1} + 1 \right) - t\right)  \right]\\
\\  & \leq \exp \left(-\frac{k^2}{2} \cdot \frac{(3/4)^2 \left( 2^k- 2^{\kappa} \right)^2 \left( \log 2 \right)^2 }{2^{k } \left( 2^k- 2^{\kappa} \right) + \frac{1}{4} 2^{k-1} \left( 2^k- 2^{\kappa} \right) \left( \log 2^{k-1} + 1 \right) }\right) 
\\  & \leq  \exp \left(- \frac{(k-1)^2 (3/4)^2 \cdot 2^{k-1} \left( 2^k- 2^{\kappa} \right) \left( \log 2 \right)^2 }{2^{k+1} \left( 2^k- 2^{\kappa} \right) + \frac{1}{4} 2^{k} \left( 2^k- 2^{\kappa} \right) \left( \log 2^{k-1} + 1 \right) }\right) 
\\& \leq \exp \left(-  \frac{(k-1) (3/4)^2 \left(\log 2 \right)^2 }{\frac{4}{k-1} + \frac{1}{2}  \left( \log 2 + \frac{1}{k-1} \right)}\right).
\end{align*}
For $k \geq 2 \cdot 10^2 + 1$, the above yields
\begin{align*}
&\Pb \left[ \frac{1}{2^k}\sum_{i \in I_{k,j}\setminus I_{\kappa,\iota}} V(i) < \frac{1}{8} \left( \log 2^{k-1} + 1 \right) \right]
\\&\leq \exp \left(- (k-1) \cdot \frac{(3/4)^2 \left(\log 2 \right)^2 }{\frac{4}{2 \cdot 10^2} + \frac{1}{2}  \left( \log 2 + \frac{1}{2 \cdot 10^2} \right)}\right).  
\end{align*}

We now let 
\[
K = \exp \left( \frac{(3/4)^2 \left(\log 2 \right)^2 }{\frac{4}{2 \cdot 10^2} + \frac{1}{2}  \left( \log 2 + \frac{1}{2 \cdot 10^2} \right)} \right),
\] 
and note that we then have that for $k \geq 2 \cdot 10^2 + 1$, 
\[ 
\Pb \left[ \frac{1}{2^k}\sum_{i \in I_{k,j}\setminus I_{\kappa,\iota}} V(i) < \frac{1}{8} \left( \log 2^{k-1} + 1 \right) \right] \leq \left( \frac{1}{K}\right)^{k-1}.
\]

Since there are less than $2^{k}$ groups $I_{\kappa,\iota}$ such that $I_{\kappa,\iota}\subset I_{k,j}$,  a union bound gives us that the probability that the bidder prefers a different option other than $L_{k,j}$ is upper bounded by $2 \cdot \left(\frac{2}{K}\right)^{k-1}$. A direct calculation shows that $2 \cdot \left(\frac{2}{K}\right)^{k-1} \leq 10^{-3}$. 
\end{proof}

We are now ready to prove Lemma~\ref{lemma:loglog}.

\begin{proof}[Proof of Lemma~\ref{lemma:loglog}]
The proof of Lemma~\ref{clm: inter_ex2} directly implies that the revenue of the considered mechanism is lower-bounded by
\begin{align*}
\frac{(1-10^{-3}) \log 2}{8} \left( \sum_{k \geq 2 \cdot 10^2 + 1}^{m-1} \frac{k}{k(k+1)} + \frac{m}{m} \right) 
&= \Omega \left( \log m \right)
\\&= \Omega \left( \log \log n \right),
\end{align*}
as a bidder who receives signal $s_{k,j}$ picks option $L_{k,j}$ with price $\frac{\log 2}{8} k$ with probability at least $1-10^{-3}$. 

The revenue of the best mechanism is upper-bounded by the optimal revenue the seller could obtain if she knew the realization of the signal. When facing signal $s_{k,j}$, the bidder's posterior is that the item type is taken uniformly at random from group $I_{k,j}$. By~\citet{BILW14,DPT16}, the better of item-type pricing and item-type bundling (conditioning now on the realization of the signal) yields a constant approximation to the optimal revenue. The revenue from item-type pricing is clearly $1$, and the revenue from item-type bundling is $O \left( \log 2^k \right)$ by Lemma~\ref{lem: ER_property} as setting $P > 6 \log 2^k$ yields constant revenue while setting $P \leq 6 \log 2^k$ yields $O \left( \log 2^k \right)$. Therefore, the optimal revenue conditional on the signal being $s_{k,j}$ must be $O \left( \log 2^k \right) = O \left( k \right)$, and the optimal (unconditional) revenue is therefore 
\[
 O \left( \sum_{k=1}^{m-1} \frac{k}{k(k+1)} + \frac{m}{m} \right) = O \left( \log m \right) = O \left( \log \log n \right).\qedhere
\]
\end{proof}

\section{Proof of Lemma~\ref{LEM:FULL_REVELATION_STRATEGIC}}~\label{app:strategic}
To compare $U(\vec V, S^*)$ and $U(\vec V, S)$, in the following we slightly abuse notation and let $\vec \pi_X$ denote the posterior distribution the buyer forms over $\itemtype$ when the signaling scheme is~$S^*$ and the buyer receives signal~$X$. 
	
Consider a lottery of the form given in Section~\ref{sec:charac} and discussed in Appendix~\ref{appx:charac}. Suppose the lottery has $l+1$ options, denoted by $L_0$, $L_1$ to $L_l$, where $L_0$ is a dummy option with price $0$ and allocation $\vec{Z_{\itemtype}}(0) = \vec{0}$ added to guarantee IR (as in Appendix~\ref{appx:charac}). Further, let $z_{\itemtype}(k)$ denote the probability with which $L_k$ allocates item of type~$i$, and $c_{\itemtype}(k)$ the price at which $L_k$ sells item of type~$\itemtype$. The expected utility of the bidder when he has valuation $\vec V$ and signal $S(X)$ is given by
\[
U(\vec V, S) = \Ex[X] {\Ex[\sig \sim S(X)] {\max_k \sum_{\itemtype} \pi_{\sig}(\itemtype) \left(V(\itemtype) z_{\itemtype}(k) - c_{\itemtype}(k) \right) }}
\]
On the other hand, if the data provider fully reveals his information, the bidder possessing this information and with value $\vec V$ would have utility 
\begin{align*}
U(\vec V, S^*) = \Ex[X]{ \max_k \sum_{\itemtype} \pi_{X}(\itemtype) (V(\itemtype) z_{\itemtype}(k) - c_{\itemtype}(k) ) }.
\end{align*}

Since the bidder's posterior when observing the realization of $\sig$ is obtained via Bayes update, we have $\vec \pi_{\sig} = \Ex[\tilde X \given \sig]{\vec \pi_{\tilde X}}$, where on the right hand side the expectation is taken over $\tilde X$, the buyer's belief of the data provider's information, drawn from the conditional distribution given the received signal~$\sig$.  
Therefore, we can write
\begin{align*}
	&U(\vec V, S)
\\& = \Ex[X] {\Ex[\sig \sim S(X)] {\max_k \Ex[\tilde X \given \sig] {\sum_{\itemtype} \pi_{\tilde X}(\itemtype) (V(\itemtype) z_{\itemtype}(k) - c_{\itemtype}(k))}} } \\
& \leq \Ex[X] {\Ex[\sig \sim S(X)] {\Ex[\tilde X \given \sig] {\max_k \sum_{\itemtype} \pi_{\tilde X}(\itemtype) (V(\itemtype) z_{\itemtype}(k) - c_{\itemtype}(k))}} } \\
& = \Ex[\tilde{X}] {\max_k \sum_{\itemtype} \pi_{\tilde{X}}(\itemtype)(V(\itemtype) z_{\itemtype}(k) - c_{\itemtype}(k)) }\\
& = \Ex[X] {\max_k \sum_{\itemtype} \pi_{X}(\itemtype)(V(\itemtype) z_{\itemtype}(k) - c_{\itemtype}(k)) } \\
&= U(\vec V,S^*),
\end{align*}
where the inequality follows from Jensen's inequality. Conditional on the signal being $s$, the distributions of $\tilde{X}|s$ and $X|s$ are identical by definition of Bayes update, which in turn directly implies the distributions of $\tilde{X}$ and $X$ are identical, and the second-to-last equality holds. This concludes the proof. 
\end{document}